\documentclass[lettersize,journal]{IEEEtran}
\usepackage{array}
\usepackage{textcomp}
\usepackage{stfloats}
\usepackage{url}
\usepackage{verbatim}
\usepackage{graphicx}
\usepackage{amsmath,amssymb,amsfonts,amsthm}
\usepackage{subcaption}
\usepackage{color}
\usepackage[linesnumbered,ruled,vlined]{algorithm2e}
\DeclareMathOperator*{\argmax}{arg\,max}
\DeclareMathOperator*{\argmin}{arg\,min}
\usepackage{algpseudocode}

\newtheorem{theorem}{Theorem}

\newtheorem{definition}{Definition}
\newtheorem{lemma}{Lemma}

\newtheorem{assumption}{Assumption}

\def\blue{\color{black}}

\def\violet{\color{black}}
\usepackage{tcolorbox}
\usepackage{lipsum}
\usepackage{hyperref}
\usepackage[noadjust]{cite}
\usepackage{enumitem}
\usepackage{booktabs}
\def\BibTeX{{\rm B\kern-.05em{\sc i\kern-.025em b}\kern-.08em
    T\kern-.1667em\lower.7ex\hbox{E}\kern-.125emX}}

\allowdisplaybreaks

\colorlet{blue}{blue}
\begin{document}
\title{Computation and Communication Co-scheduling for Multi-Task Remote Inference\\ %at the Wireless Edge\\
}

%\author{\IEEEauthorblockN{ }
\author{Md Kamran Chowdhury Shisher,~\IEEEmembership{Member,~IEEE,}
        Adam Piaseczny,~\IEEEmembership{Student~Member,~IEEE,}\\
Yin~Sun,~\IEEEmembership{Senior~Member,~IEEE,} Christopher G. Brinton,~\IEEEmembership{Senior~Member,~IEEE}\IEEEcompsocitemizethanks{\IEEEcompsocthanksitem This paper was presented in part at IEEE INFOCOM, 2025 \cite{ShisherINFOCOM2025}. 

M.K.C. Shisher, A. Piaseczny, and C. Brinton are with the Elmore Family School of Electrical and Computer Engineering, Purdue University, West Lafayette, IN, 47907, USA (e-mail: mshisher@purdue.edu, apiasecz@purdue.edu, cgb@purdue.edu). 
        
        Y. Sun is with the Department of Electrical and Computer Engineering, Auburn University, Auburn, AL,
36849, USA (e-mail: yzs0078@auburn.edu). 
 
Y. Sun was supported in part by the National Science Foundation (NSF) under grant CNS-2239677. M. K. C. Shisher, A. Piaseczny, and C. Brinton were supported in part by the Office of Naval Research (ONR) under grants N00014-23-C-1016 and N00014-22-1-2305, and by NSF CPS-2313109.}
}
%\IEEEauthorblockA{\textit{Anonymous)} \\
%\textit{Anonymous)}\\}
%\and
%\and
%\and
%\and
%\and

\maketitle
\begin{abstract}
In multi-task remote inference systems, an intelligent receiver (e.g., command center) performs multiple inference tasks (e.g., target detection) using data features received from several remote sources (e.g., edge {\violet devices}). Key challenges to facilitating timely inference in these systems arise from (i) limited computational power of the sources to produce features from their inputs, and (ii) limited communication resources of the channels to carry simultaneous feature transmissions to the receiver. We develop a novel computation and communication co-scheduling methodology which determines feature generation and transmission scheduling to minimize inference errors subject to these resource constraints. Specifically, we formulate the co-scheduling problem as a weakly-coupled Markov decision process with Age of Information (AoI)-based timeliness gauging the inference errors. To overcome its PSPACE-hard complexity, we analyze a Lagrangian relaxation of the problem, which yields gain indices assessing the improvement in inference error for each potential feature generation-transmission scheduling action. Based on this, we develop a reoptimized maximum gain first (MGF) policy. We show that this policy is asymptotically optimal for the original problem as the number of inference tasks and the available communication and computation resources increase, provided the ratio among them remains fixed. Experiments demonstrate that reoptimized MGF obtains significant improvements over baseline policies for varying numbers of tasks, channels, and sources.
%In this paper, we investigate a multi-task remote inference system where an edge receiver (e.g., a command center) performs multiple inference tasks (e.g., classification, detection, and segmentation) using features delivered from resource-constrained remote sources (e.g., UAVs). The inference error for each task is determined by the timeliness of the delivered features, where we employ the Age of Information (AoI) as the metric of timeliness. Due to limited computational resources and communication channels, the delivered features may not be timely. To minimize inference errors, we study a computation and communication co-scheduling problem to decide which features to compute and which features to transmit. The computation and communication co-scheduling problem is a weakly-coupled Markov decision process that is PSPACE hard. For this setting, we develop a maximum gain first policy. We demonstrate the asymptotic optimality of our policy as the number of inference tasks increases. Additionally, we show the potential of our policy through simulation results.
\end{abstract}
\section{Introduction}
The simultaneous advances in machine learning and communication technologies have spurred demand for intelligent networked systems across many domains \cite{giordani2020toward, akter2022iomt}. These systems, whether for commercial or military purposes, often rely on timely information delivery to a remote receiver for conducting several concurrent decision-making and control tasks \cite{shishertimely}. For example, consider intelligence, surveillance, and reconnaissance (ISR) \cite{peterson2020persistent} objectives within military operations. A command center may employ signals transmitted from several dispersed military assets, e.g., unmanned aerial vehicles (UAVs), to simultaneously classify friendly versus hostile agents, track the positions of targets, and detect anomalous sensor data. Similarly, in intelligent transportation \cite{chellapandi2023federated}, near real-time prediction of road conditions, vehicle trajectories, and other tasks is crucial for traffic management and safety. Smart retail management system requires to infer current inventory and to classify customer reactions. 

As the number and complexity of learning tasks in such applications continues to rise, there are two salient challenges to facilitating timely multi-task remote inference (MTRI). First, \textit{there are limited wireless resources (e.g., orthogonal frequency channels) available for information transmission from sources to the receiver at the network edge}. The observed information at the sources, often edge devices, can be high dimensional and transmitting the high dimensional data to cloud server requires a significant amount of communication resources. Due to limited communication resources, the delivered information may end up being stale. Because of advancement of hardware at the edge devices, questions arise as to whether this receiver (e.g., server)-only processing is efficient. Edge device computing has a potential to reduce communication resources needed. Instead of sending high dimensional signal values, AI-powered edge devices may locally construct low-dimensional feature representations of their high-dimensional signal observations (e.g., video streams) to send in lieu of the raw measurements.
This may be developed, for example, by splitting the neural network for each task at a designated cut layer, and implementing the two parts at the edge device and receiver, respectively \cite{wu2023split}. However, this also leads to the second challenge: \textit{the sources, often edge devices, have heterogeneous on-board computational capabilities, limiting their ability to simultaneously construct multiple features required by different tasks}. For example, smart glasses may have low compute power, whereas high compute power can be installed on a vehicles. Hence, the feature computational ability of smart glasses is significantly lower compared to vehicles. 

Due to these resource limitations, the features at the receiver may not always reflect the freshest source information. It is thus critical to ascertain which tasks require feature updates most urgently at any given time, i.e., to determine where to focus available MTRI resources. \emph{Age of Information} (AoI), introduced in \cite{song1990performance, kaul2012real}, can provide a useful measure of information freshness of the receiver. Specifically, consider packets sent from a source to a receiver: if $U(t)$ is the generation time of the most recently received packet by time $t$, then the AoI at time $t$ is the difference between $t$ and $U(t)$. Recent works on remote inference \cite{ShisherMobihoc, shishertimely, shisher2023learning,  shisher2024AR} have shown that the inference errors for different tasks can be expressed as functions of AoI, and that surprisingly, these functions are not always monotonic. Additionally, AoI can be readily tracked in an MTRI system on a per-task basis, making it a promising metric for determining how to prioritize resource allocation. Motivated by this, we pose the following research question:
\begin{quote}
\textbf{\textit{How can we develop a computation and communication co-scheduling methodology for MTRI systems that leverages AoI indicators of timeliness to minimize the inference errors across tasks while adhering to network resource constraints?}}
\end{quote}

\subsection{Outline and Summary of Contributions}
%Ensuring timely feature updates to minimize inference errors in MTRI systems is a challenging problem due to limited computational and communication resources. In this paper, we address this challenging problem by developing a novel a communication-computation co-scheduling methodology for timely multi-task remote inference.
%In answering this question, we make the following contributions:
\begin{itemize}[leftmargin=4mm]
    \item We formulate the MTRI policy optimization problem to minimize discounted infinite horizon inference errors subject to source feature computation and transmission constraints (Sec.~\ref{sec:system}\&\ref{ssec:formulation}). This optimization considers the dependency of the inference error on AoI measures for each task's features and their impact on the prediction results.
    %scheduler determines which features to compute at the sources and which features to transmit to the receiver at each timeslot, 
    
    \item We show how the co-scheduling problem can be modeled as a weakly-coupled Markov Decision Process (MDP) (Sec.~\ref{WeaklyMDP}). Weakly-coupled MDPs are extensions of restless bandits by allowing for multiple resource constraints. To overcome the associated PSPACE-hard complexity, we derive a Lagrangian relaxation of the original problem, and establish its optimal decision (Sec.~\ref{sec:relaxation}, Lemma \ref{lemma1}). Analyzing the dual problem allows us to obtain a \emph{gain index} for each task, which quantifies the reduction in inference error from scheduling it.
    
    \item Leveraging these gain indices, we propose a novel maximum gain first (MGF) policy (Algorithm \ref{alg:gain}) to solve the original problem, iteratively scheduling features/tasks with maximum gain until capacity is reached (Sec.~\ref{ssec:mgf}). The MGF policy is a special case of the re-optimized fluid (ROF) policy introduced in \cite{brown2023fluid} for general weakly coupled MDPs. We prove that in the MTRI problem, our MGF policy achieves asymptotic optimality at a rate of $O (\frac{1}{\sum_{m=1}^M \sqrt{rk_m}})$, where $rk_m$ is the number of inference tasks per source $m$ and $M$ is the total number of sources (Theorem \ref{theorem4}, Sec.~\ref{ssec:analysis}). Notably, this optimality gap is tighter than the $O (\frac{1}{\sqrt{\sum_{m=1}^M rk_m}})$ bound established in \cite{brown2023fluid}.
    Our scheduling results are applicable to any bounded penalty functions of AoI with multiple resource constraints. We also provide Algorithm \ref{alg:gain1} by reducing the number of optimization variables of Algorithm \ref{alg:gain}. 

    \item We conduct numerical experiments to demonstrate our policy on synthetic and real-world inference tasks (Sec.~\ref{sec:simulation}). For synthetic evaluations, we use three different types of inference error functions that are widely used in AoI literature \cite{Tripathi2019, SunNonlinear2019}. For real-world inference tasks, we consider remote robot car detection and vehicular inference tasks (image segmentation and traffic prediction). In the remote robot car detection experiment, $4$ sources and $5$ robot cars are used, where source 1 observes 2 of the cars, while the remaining 3 sources each monitor a single car. In the vehicular inference tasks, roadside sensors equipped with cameras are used as the MTRI sources. We find that MGF significantly outperforms baseline policies in terms of cumulative errors as the number of tasks, channels, and sources are varied. A widening margin is observed as the number of tasks increases, consistent with our optimality analysis.
    %particularly as the number of tasks and number of sources to show the potential of our proposed policy by using two inference tasks per source (image segmentation and traffic prediction). In this simulation, road side units equipped with cameras are used as information sources.
\end{itemize}

\subsection{Related Works}
\label{ssec:related}
The concept of Age of Information (AoI) has attracted significant research efforts; see, e.g., \cite{shisher2021age,kaul2012real,sun2017update, yates2015lazy, KamKompellaEphremidesTIT, kadota2018optimizing, soleymani2016optimal, SunSPAWC2018, SunNonlinear2019, chen2021uncertainty, wang2022framework, YinUpdateInfocom, orneeTON2021, Tripathi2019, klugel2019aoi, bedewy2021optimal,hsu2018age, sun2019closed, Kadota2018, ornee2023whittle, pan2023sampling, sun2023age, sun2022age, OrneeMILCOM, ayan2023optimal} and a survey \cite{yates2021age}. Initially, research efforts were centered on analyzing and optimizing linear functions of AoI, as a performance metric of communication networks \cite{kaul2012real, sun2017update, yates2015lazy, KamKompellaEphremidesTIT, kadota2018optimizing}. Recently,  researchers have shifted their efforts towards optimizing the performance of real-time applications, such as remote estimation \cite{SunTIT2020, orneeTON2021, OrneeMILCOM, orneeTON2021, ornee2023whittle}, remote inference \cite{shisher2021age, ShisherMobihoc, shishertimely, shisher2025ICC}, and control systems \cite{klugel2019aoi, ayan2023optimal}, by leveraging AoI as a tool. Previous works \cite{shisher2021age, ShisherMobihoc, shishertimely} have demonstrated that the performance of remote inference systems depends on the AoI of the features they utilize; specifically, representing inference error as a function of AoI. In this paper, we consider the more challenging MTRI case with multiple information sources, an edge receiver, and multiple inference tasks for each source. Motivated by the prior work, we consider the dependency of inference error for each task on the AoI of features delivered to the receiver. Notably, the inference error function in our case can be monotonic or non-monotonic with AoI. This paper is also related to the field of signal-agnostic remote estimation. The prior studies \cite{SunNonlinear2019, orneeTON2021, ornee2023whittle, Ornee2021performance, SunTIT2020, klugel2019aoi, pan2023sampling, Tripathi2019} in signal-agnostic remote estimation focused on Gaussian and Markovian processes and found that the estimation error can be represented as a function of AoI values.

%AoI has become a widely-used metric in the analysis and optimization of systems including communication networks \cite{Kadota2018, Kadota2019}, control systems \cite{klugel2019aoi, champati2019performance}, remote estimation \cite{SunTIT2020, orneeTON2021, OrneeMILCOM}, and remote inference \cite{shisher2021age, ShisherMobihoc, shishertimely}. Previous works \cite{shisher2021age, ShisherMobihoc, shishertimely} have demonstrated that the performance of remote inference systems depends on the AoI of the features they utilize; specifically, representing inference error as a function of AoI. In this paper, we consider the more challenging MTRI case with multiple information sources, an edge receiver, and multiple inference tasks for each source. Motivated by the prior work, we consider the dependency of inference error for each task on the AoI of features delivered to the receiver. Notably, the inference error function in our case can be monotonic or non-monotonic with AoI. 

Researchers have explored scheduling policies to minimize linear and non-linear functions of AoI in multi-source networked intelligent systems \cite{Tripathi2019, kadota2018optimizing, hsu2018age, ornee2023whittle, sun2019closed, Kadota2018, chen2021uncertainty,chen2023index, shisher2023learning, shishertimely, OrneeMILCOM}. Early studies focused on systems with limited communication resources and binary actions for each source \cite{Tripathi2019, kadota2018optimizing, hsu2018age, ornee2023whittle, sun2019closed, Kadota2018, chen2021uncertainty,chen2023index, OrneeMILCOM}. More recent research has expanded to consider scenarios with multiple actions per source \cite{shisher2023learning, shishertimely}. These scheduling problems have been formulated as restless multi-armed bandit (RMAB) problems, with either binary or multiple actions. While RMABs are weakly coupled MDPs, which are in general PSPACE-hard, Whittle index \cite{Tripathi2019, kadota2018optimizing, hsu2018age, ornee2023whittle, sun2019closed, Kadota2018, chen2021uncertainty, shishertimely} and gain index \cite{chen2023index, OrneeMILCOM, shisher2023learning} approaches have been shown to yield asymptotically optimal policies under certain conditions, notably the global attractor condition \cite{whittle1988restless, shishertimely, gast2021lp}. However, these previous works have not addressed the presence of computation resource constraints and multiple inference tasks characteristic of MTRI systems. By considering these factors, our MTRI computation and communication co-scheduling problem becomes a weakly coupled MDP that is more general than RMAB and requires new approaches to solve it.

\begin{figure*}[h]
  \centering
  
  \begin{subfigure}[t]{0.42\textwidth}
\includegraphics[width=\textwidth]{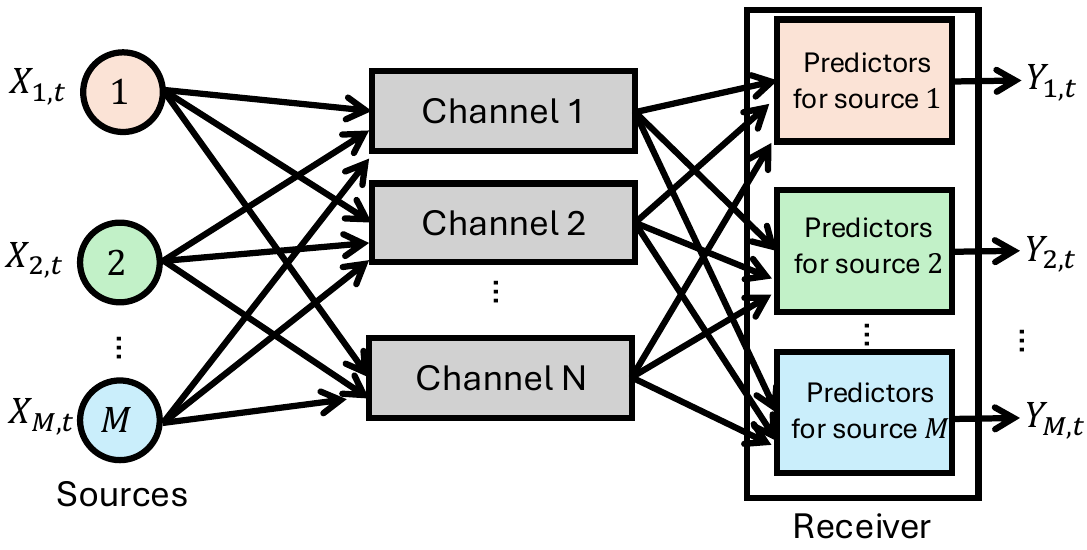}
  \subcaption{An MTRI system.}
\end{subfigure}
% <— this is important. There should be no empty line here. 
%
\hspace{10mm} 
\begin{subfigure}[t]{0.45\textwidth}
\includegraphics[width=\textwidth]{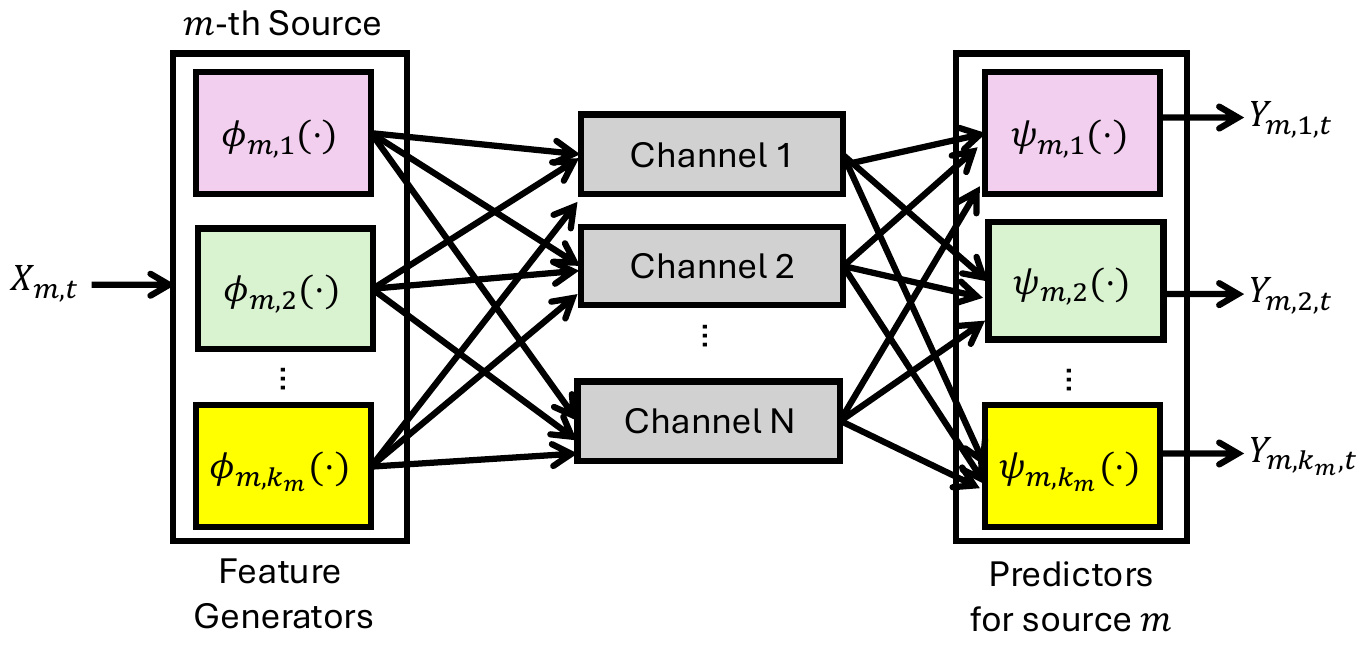}
\subcaption{\small $m$-th source with $k_m$ inference tasks.}
\end{subfigure}
\vspace{-0.05in}
\caption{\violet \small System Model: (a) The multi-task remote inference (MTRI) system, where $M$ sources are connected to an intelligent receiver via $N$ wireless channels. Each source $m$ observes the high-dimensional signal $X_{m,t}$. At each time slot $t$, the goal of the receiver is to predict time-varying targets $(Y_{1, t}, Y_{2, t}, \ldots, Y_{M,t})$, where the target $Y_{m,t}$ is a tuple of $k_m$ time-varying signals $(Y_{m, 1, t}, Y_{m, 2, t}, \ldots, Y_{m, k_m, t})$ of source $m$. (b) To address the $k_m$ inference tasks, each source $m$ generates and sends $k_m$ features, each associated with a single inference task. The receiver uses $k_m$ predictors to perform the $k_m$ inference tasks of source $m$, where $j$-th predictor for source $m$ predicts the target $Y_{m, j, t}$ using the most recently delivered feature from the $j$-th feature generator of source $m$.\label{fig:systemmodel}}
\vspace{-0.2in}
\end{figure*}

Recently, a few works \cite{brown2023fluid, gast2024reoptimization} %{\color{red} Yin: journal references should have volume, issue, and page numbers.}
have developed re-optimized fluid policies which are asymptotically optimal for general weakly-coupled MDPs, using linear programming solutions.
%investigated asymptotically optimal policies for general weakly-coupled MDPs, using linear programming solutions.
Our work builds upon the approach provided in \cite{brown2023fluid} to develop scheduling policies for MTRI systems with multiple sources, channels, and inference tasks, which we also show are asymptotically optimal. Importantly, the optimality gap obtained in our paper is tighter than the bound established in \cite{brown2023fluid}. Beyond minimizing inference errors, our gain indicies-based policy is more generally applicable to the minimization of any bounded penalty function of AoI which involves multiple actions per source/task and multiple resource constraints.

\section{System Model}
\label{sec:system}

% \begin{figure}[t]
% \centering
% \includegraphics[width=0.40\textwidth]{SystemModel2.eps}
% \vspace{-2mm}
% \caption{\small An MTRI system with $M$ sources, $N$ channels, $k_m$ feature generators per source $m$ (in this figure, $k_m=2$), and $K$ predictors. 
% %Each source $m$ generates $k_m$ features, and transmission to the receiver occurs over $N$ wireless channels. 
% %Due to limited computation and communication resources, a scheduler decides which features to generate and transmit to the receiver. The receiver utilizes most recently delivered features to infer $K=\sum_{m=1}k_m$ time-varying targets.
% \label{fig:scheduling}
% %{\color{red} Yin: the lower part of figure does not look good.}
% }
% \vspace{-4mm}
% \end{figure}

\subsection{Overview}
{\violet We consider the multi-task remote inference (MTRI) system, as illustrated in Fig.~\ref{fig:systemmodel}(a), where $M$ sources are connected to an intelligent receiver via $N$ wireless channels and the receiver. Each source $m$ can be equipped with one sensor or multiple sensors addressing multiple tasks as depicted in Fig.~\ref{fig:systemmodel}(b); for example, in an ISR system, an UAV equipped with a camera can act as a source, transmitting processed video frames to a central command center for multiple tasks such as object recognition and anomaly detection. %On the other hand, a single source can have multiple sensors, with each task uniquely associated with a specific sensor; for example, 
Another example is a robot using a LIDAR sensor for 3D mapping and a chemical sensor for hazardous material detection. At each time slot $t$, each source $m$ observes a high dimensional time-varying signal $X_{m,t} \in \mathcal{X}_m$, where $\mathcal{X}_m$ represents the set of possible observations, e.g., possible values of video frames captured by a camera. Sources will progressively generate low-dimensional feature representations of their observations for communication-efficient transmission over the $N$ wireless channels when they are scheduled.}

%each source $m$ observes a time-varying multi-dimensional signal $X_{m, t}$, generating low-dimensional feature representations for efficient transmission over the wireless channels when scheduled.

%We consider the multi-task remote inference (MTRI) system in Fig.~\ref{fig:scheduling}. $M$ sources are connected to an intelligent receiver via $N$ wireless channels. For example,
%(e.g., a command center)
%in an ISR system, UAVs equipped with cameras can act as sources, transmitting processed video frames to a central command center for further analysis. 
%At every time slot $t$, each source $m$ observes a time-varying signal $X_{m,t} \in \mathcal{X}_m$, where $\mathcal{X}_m$ represents the set of possible observations, e.g., possible values of video frames captured by a camera. Sources will progressively generate low-dimensional feature representations of their observations for communication-efficient transmission over the $N$ wireless channels when they are scheduled.

At each time $t$, the receiver employs multiple predictors trained to infer targets based on received source features. Specifically, for each source $m$, the receiver aims to infer 
$Y_{m, t}$ which is a tuple of $k_m$ time-varying targets $(Y_{m, 1, t}, Y_{m, 2, t}, \ldots, Y_{m, k_m, t})$ of source $m$. These targets can represent various inference tasks, e.g., object detection, segmentation, anomaly detection, 3D mapping, depending on the nature of the observations and {\blue the goals of the MTRI system}. In the system, there are a total of $K=\sum_{m=1}^Mk_m$ inference tasks. The tuple $(m, j)$ uniquely identifies the $j$-th inference task of source $m$.

%We consider a multi-source, multi-task remote inference system composed of $M$ sources (e.g., UAVs, robots, vehicles) and a receiver (e.g., command center), as illustrated in Fig.~\ref{fig:scheduling}. Each source $m$ observes a time-varying signal $X_{m, t} \in \mathcal{X}_m$, progressively generates low-dimensional representations or features from the observation $X_{m, t}$, and transmits the features to the receiver. The features from all $M$ sources are transmitted through $N$ shared communication channels. At every time slot $t$, the receiver 
%utilized multiple predictors to infer multiple time-varying targets $(Y_{m, 1, t}, Y_{m, 2, t}, \ldots, Y_{m, k_m, t})$ (e.g. classification, detection, segmentation) related to the observation $X_{m, t} \in \mathcal{X}_m$ of source $m$.

%performs inference on 
%multiple time-varying targets $(Y_{m, 1, t}, Y_{m, 2, t}, \ldots, Y_{m, k_m, t})$ (e.g. classification, detection, segmentation) related to the observation $X_{m, t} \in \mathcal{X}_m$ of source $m$. 

\subsection{Computation Model}
Each source $m$ is equipped with $k_m$ {\blue pre-trained} feature generators. The $j$-th feature generator of source $m$, designed for the $(m, j)$-th inference task, is denoted by a function $\phi_{m,j}:\mathcal X_m \mapsto \mathcal Z_{m, j}$. This function takes the observation $X_{m,t} \in \mathcal X_{m}$ as input and generates a feature $\phi_{m,j}(X_{m,t}) \in \mathcal Z_{m, j}$, where $\mathcal Z_{m, j}$ is the set of possible features generated by $\phi_{m,j}(\cdot)$. 
%At any time slot $t$, this feature generation process consumes $c_{m, j}$ processing units (e.g., central processing unit or graphics processing unit) of source $m$. The generated feature, along with its time stamp, is then stored in the source $m$. 
%replacing any older feature previously stored there. {\red What should we say instead of buffer?} 
To account for computational resource limitations, we assume it is not feasible to activate all feature generators at every time slot. Specifically, for source $m$, at most $C_m$ feature generators can be activated at any given time. 

\subsection{Communication Model}
As illustrated in Fig.~\ref{fig:systemmodel}(a), $N$ wireless channels are shared among the $M$ sources. If scheduled at time $t$, the $(m, j)$-th feature generator produces $\phi_{m,j}(X_{m,t})$ and transmits to the receiver using $n_{m, j}$ channels. For simplicity, we assume perfect channels, i.e., features sent at time slot $t$ are delivered error-free at time slot $t+1$. However, our results can be extended to accommodate erasure channels, where data loss may occur.

Due to the limited number of channels, at any given time $t$, only features for a subset of inference tasks can be transmitted. Consequently, the receiver may not have fresh features for all tasks. If the most recently delivered feature for the $(m, j)$-th inference task was generated $\Delta_{m, j}(t)$ time slots ago, then the feature at the receiver is $\phi_{m,j}(X_{m, t-\Delta_{m, j}(t)}),$
where $\Delta_{m, j}(t)$ is its age of information (AoI) \cite{kaul2012real, shishertimely}. Let $U_{m, j}(t)$ be the generation time of the most recent delivered feature. Then, the AoI can be formally defined as:
\begin{align}\label{AOI}
\Delta_{m, j}(t):=t-U_{m, j}(t),
\end{align}
which is the difference between the current time $t$ and the generation time $U_{m, j}(t)$.

\subsection{Inference Model}
The receiver is equipped with $K$ {\blue pre-trained} predictors, where $\psi_{m, j}:\mathcal Z_{m, j} \times \mathbb Z^{+} \mapsto \mathcal Y_{m, j}$ is the predictor function for the $(m,j)$-th inference task. Specifically, predictor $\psi_{m, j}(\cdot, \cdot)$ takes the most recently delivered feature $\phi_{m,j}(X_{m, t-\Delta_{m, j}(t)}) \in \mathcal Z_{m, j}$ and its AoI $\Delta_{m, j}(t) \in \mathbb Z^{+}$ as inputs and generates the predicted result $\hat Y_{m, j, t} \in \mathcal Y_{m, j}$. In other words, we assume the predictor may in general adjust/calibrate the inference based on the AoI. 

%In this paper, we focus on signal-agnostic scheduling policies, where scheduling decisions are made without knowledge of the observed process's signal value. Under signal-agnostic scheduling, the processes $\{(Y_{m, j, t}, X_{m,t}), t = 0, 1, \ldots\}$ and $\{\Delta_{m, j}(t), t = 0, 1, \ldots\}$ are independent \cite{shishertimely, ShisherMobihoc, champati2019performance}. 
%Moreover, signal-agnostic policies for stationary processes serve as a valuable foundation for studying non-stationary processes \cite{tripathi2021online}.
We make the following assumptions:
\begin{assumption}\label{independence}
The processes $\{(Y_{m, j, t}, X_{m,t}), t = 0, 1, \ldots\}$ and $\{\Delta_{m, j}(t), t=0, 1, \ldots\}$ are independent for all $(m, j)$.
\end{assumption}

\begin{assumption}\label{stationary}
    The process $\{(Y_{m, j, t}, X_{m,t}), t = 0, 1, \ldots\}$ is stationary for all $(m, j)$, i.e., the joint distribution of $(Y_{m, j, t}, X_{m,t-k})$ does not change over time $t$ for all $k\geq 0$. %Exploring the scheduling results for non-stationaryprocesses is an important research direction. The results presented in our paper for stationary process can serve as a foundation of studying the scheduling policy for non-stationary processes. 
\end{assumption}

%{\color{red} Yin: You used the joint distribution for different age values, not just age = 0.}

{\blue Assumption \ref{independence} is satisfied for signal-agnostic scheduling policies in which the scheduling decisions are made based on AoI and the distribution of the process, but not on the values taken by the process \cite{shishertimely}. Assumption \ref{stationary} is utilized to ensure that the inference error is a time-invariant function of the AoI, as we will see in \eqref{instantaneous_err1}. It is practical to approximate time-varying functions as time-invariant functions in the scheduler design. Moreover, the scheduling policy developed for time-invariant AoI functions serves as a valuable foundation for studying time-varying AoI functions \cite{tripathi2021online}.}

%In this paper, we will study signal-agnostic scheduling policies (see Section \ref{scheduling}) in which scheduling decisions are determined without using the knowledge of the signal value of the observed process. If the scheduler is signal-agnostic, then the processes $\{(Y_{m, j, t}, X_{m,t}), t=0, 1, 2, \ldots\}$ and $\{\Delta_{m, j}(t), t=0, 1, 2, \ldots\}$ are independent \cite{shishertimely, ShisherMobihoc}. Moreover, the signal-agnostic scheduling policy for stationary process can be a foundation to study non-stationary process \cite{tripathi2021online}. 

Under Assumptions \ref{independence}-\ref{stationary}, given an AoI $\Delta_{m, j}(t)=\delta$, the inference error for the $(m, j)$-th inference task at time slot $t$ can be represented as a function of AoI $\delta$ \cite{shishertimely, ShisherMobihoc}:
\begin{align}\label{instantaneous_err1} 
&p_{m, j}(\delta) \nonumber \\ &=\mathbb E_{Y, X \sim P_{Y_{m, j, t}, X_{m, t-\delta}}}\bigg[L_{m, j}(Y,\psi_{m, j}(\phi_{m, j}(X), \delta))\bigg],
\end{align}
where $P_{Y_{m, j, t}, X_{m, t-\delta}}$ is the joint distribution of the target $Y_{m, j, t}$ and the observation $X_{m, t-\delta}$, and $L_{m, j}(y, \hat y)$ is the loss function for the task that measures the loss incurred when the actual target is $y$ and the inference result is $\hat y$ (e.g., cross-entropy loss for a classification task). 

\section{Scheduling Problem Formulation}\label{scheduling}
%In this section, we will define the computation and communication co-scheduling policy and formulate the corresponding scheduling optimization for MTRI systems.

\subsection{Scheduling Policy and Optimization}
\label{ssec:formulation}
We denote the scheduling policy as
$$\pi = (\pi_{m,j}(0), \pi_{m,j}(1), \ldots)_{\forall (m, j)},$$
where $\pi_{m,j}(t) \in \{0, 1\}$. At time slot $t$, if $\pi_{m,j}(t) = 1$, the features for the $(m,j)$-th inference task are generated and transmitted to the receiver; otherwise, if $\pi_{m,j}(t) = 0$, this generation and transmission does not occur. 
We let $\Pi$ denote the set of all signal-agnostic and causal scheduling policies $\pi$ that satisfy three conditions: (i) the scheduler knows the AoIs up to the present time, i.e., $\{\Delta_{m, j}(k)\}_{\forall m, j, k\leq t}$, (ii) the scheduler does not know signal values $\{X_{m, t}, Y_{m, j, t}\}_{\forall m, j, t}$, and (iii) the scheduler has access to the inference error functions $p_{m, j}(\delta)$ for all $(m, j)$. 

Under any scheduling policy $\pi$, the AoI $\Delta_{m, j}(t)$ for each inference task $(m, j)$ evolves according to:  
\begin{align}\label{AoIprocess}
\!\! \Delta_{m,j}(t+1)=
  \begin{cases}
        1, &\text{if}~\pi_{m, j}(t)=1 \\
        \Delta_{m,j}(t)+1, &\text{otherwise}.
    \end{cases}
    \end{align} 
We assume that the initial AoI of each task $(m, j)$ is a finite constant, e.g., $\Delta_{m,j}(0)=1$. 

%Let $f=(u_{m, 1}(0), u_{m, 1}(1), \ldots)_{m=0}^M$ be the scheduling policy for \emph{feature generation scheduler }, where $u_{m, 1}(t)\in \{0, 1\}$. If $u_{m, 1}(t)=0$, the $m$-th feature generator starts processing the source signal $X_t$ using $C_m$ computation resources at time-slot $t$; otherwise, the $m$-th feature generator at time $t$ does not initiate a new feature processing. We consider non-preemptive feature processing, i.e., if $m$-th feature generator is processing a feature at time slot $t$, it can not initiate a new feature processing at the same time slot $t$. Let $\mathcal F$ be the set of all non-preemptive causal scheduling policy $f$. Next, we denote $g=(u_{m, 2}(0), u_{m, 2}(1), \ldots)_{m=0}^M$ as the scheduling policy for \emph{transmission scheduler }, where $u_{m, 2}(t)\in \{0, 1\}$.  If $u_{m, 2}(t)=1$, at the beginning of time slot $t$, the source sends feature from the $m$-th buffer to the receiver by using $N_m$ channels; otherwise, source does not send feature from the $m$-th buffer. 

Our goal is to find a policy $\pi \in \Pi$ that minimizes the infinite horizon discounted sum of inference errors over the $K$ tasks: 
\begin{align}\label{Multi-scheduling_problem}
\bar p_{opt} &=\inf_{\pi \in \Pi}\sum_{t=0}^{\infty}\frac{\gamma^t }{K}\sum_{m=1}^M \sum_{j=1}^{k_m} \mathbb{E}_{\pi} \left[ w_{m, j} p_{m, j}(\Delta_{m, j}(t))\right], \\\label{Sceduling_constraint1}
&\quad ~~\mathrm{s.t.} {\sum_{j=1}^{k_m} \pi_{m, j}(t) \leq C_m,} t=0, 1,\ldots, m=1, \ldots, M,
\\\label{Sceduling_constraint2}
&\quad \quad \quad \sum_{m=1}^M \sum_{j=1}^{k_m} \pi_{m, j}(t)n_{m, j} \leq N, ~~t=0, 1, 2, \ldots,
\end{align}
%{\color{red} Yin: keep the equality notion = in the second line, not the first line.}
where $w_{m,j}\geq 0$ is the weight (e.g., priority) associated with the $(m, j)$-th inference task, and the discount factor $0 < \gamma < 1$ quantifies the diminishing importance of an inference task over time. At most $C_m$ feature generators for source $m$ can compute features at time $t$. Transmitting features for the $(m, j)$-th inference task requires $n_{m, j}$ of the $N$ wireless channels available. For each task, its inference error, $p_{m, j}(\Delta_{m, j}(t))$, depends on its AoI $\Delta_{m, j}(t)$ at time slot $t$, indicating the freshness of the feature used for inference.

\subsection{Weakly Coupled MDP Formulation}\label{WeaklyMDP}
The problem \eqref{Multi-scheduling_problem}-\eqref{Sceduling_constraint2} is a weakly coupled Markov decision process (MDP) \cite{brown2023fluid, gast2024reoptimization, nadarajah2024self} with $K$ sub-MDPs (referred to as \textit{arms} in the bandit literature), one per inference task $(m, j)$ across sources $m$. The state of each $(m, j)$-th MDP at each time $t$ is represented by the AoI $\Delta_{m, j}(t)$. 
The action is $\pi_{m, j}(t)$, and its per-timeslot cost is $\gamma^t  p_{m, j}(\Delta_{m, j}(t))$ with discount factor $0 < \gamma < 1$. We can see that the state evolution defined in \eqref{AoIprocess} and the cost of each MDP depends only on its current state and action. However, the actions for all MDPs $(\pi_{m, j}(t))_{\forall m, j}$ need to satisfy the constraints in \eqref{Sceduling_constraint1}-\eqref{Sceduling_constraint2}. This interdependence of actions across MDPs through multiple resource constraints, despite independent state transitions and costs, makes the overall problem \eqref{Multi-scheduling_problem}-\eqref{Sceduling_constraint2} a weakly coupled MDP \cite{brown2023fluid, gast2024reoptimization, nadarajah2024self}. Weakly coupled MDPs are PSPACE-hard because the number of states and actions grow exponentially with the number of sub-MDPs.

The restless multi-armed bandit (RMAB) problem is a special case of the weakly coupled MDP in \eqref{Multi-scheduling_problem}-\eqref{Sceduling_constraint2}. RMAB considers a single resource constraint, whereas our problem involves multiple resource constraints \eqref{Sceduling_constraint1}-\eqref{Sceduling_constraint2}. The PSPACE-hard complexity of RMABs can be overcome by using Whittle indices \cite{whittle1988restless, ShisherMobihoc, shishertimely, ornee2023whittle, Tripathi2019, kadota2018scheduling, chen2021uncertainty}, gain indices \cite{shisher2023learning, OrneeMILCOM, chen2023index}, and linear programming-based indices \cite{chen2021scheduling} to construct asymptotically optimal policies, provided indexability and/or global attractor conditions are satisfied. However, these RMAB policies cannot be directly applied to our more general problem due to the presence of multiple resource constraints, which requires us to develop a new solution approach.

\noindent \textbf{Solution approach.} In Sec.~\ref{sec:relaxation}\&\ref{sec:policy}, we follow the approach depicted in Fig. \ref{fig:flowchart} to solve \eqref{Multi-scheduling_problem}-\eqref{Sceduling_constraint2}. We begin by deriving a relaxed Lagrangian problem. We then utilize the resulting solution to construct a maximum gain first (MGF) policy (Algorithm \ref{alg:gain}) for \eqref{Multi-scheduling_problem}-\eqref{Sceduling_constraint2}. Theorem \ref{theorem4} will demonstrate that the MGF policy becomes asymptotically optimal as the number of inference task $k_m$ for each source $m$ increases.

\begin{figure}[t]
\centering
\includegraphics[width=0.50\textwidth]{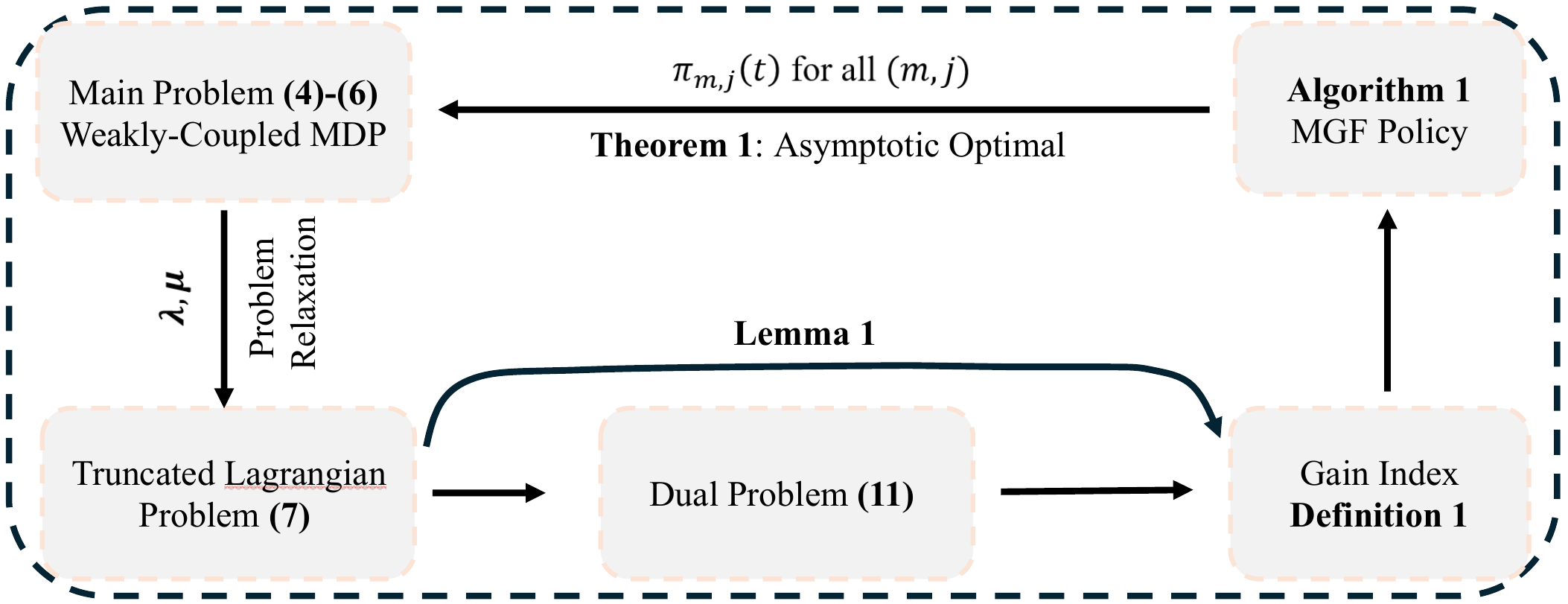}
\caption{\small Overview of our process to design the scheduling policy. %\chris{Dual problem (10)?}
\label{fig:flowchart}
}
\end{figure}

%Restless multi-armed bandit (RMAB) problem is a special case of weakly coupled MDP \eqref{Multi-scheduling_problem}-\eqref{Sceduling_constraint2}, where RMAB considers binary action for each sub-MDP $(m, j)$ and one resource constraint. RMAB problem is simpler than \eqref{Multi-scheduling_problem}-\eqref{Sceduling_constraint2}. This is because \eqref{Multi-scheduling_problem}-\eqref{Sceduling_constraint2} considers multiple actions for each task with multiple resource constraints. It is well known that RMAB is PSPACE hard. However, it is possible to construct asymptotically optimal policy for RMAB using Whittle index \cite{}, gain index \cite{}, and linear programming-based index \cite{}. The policy designed for RMAB problem in the earlier works \cite{} can not be applied to our weakly decoupled MDP \eqref{Multi-scheduling_problem}-\eqref{Sceduling_constraint2}. }

% Comments for additional experiments
% One device -> multiple sensors -> multiple tasks
% Camera -> segmentation/path prediction & object detection.
% Other sensors? -> take a look

\section{Problem Relaxation}
\label{sec:relaxation}

\subsection{Lagrangian Relaxation and Dual Problem}
To develop an asymptotically optimal policy for \eqref{Multi-scheduling_problem}-\eqref{Sceduling_constraint2}, following recent techniques for weakly coupled MDPs \cite{nadarajah2024self, brown2023fluid}, we first relax the problem using Lagrange multipliers. We associate a vector of non-negative Lagrange multipliers $\boldsymbol{\lambda}_{t}=(\lambda_{1,t},\lambda_{2,t}, \ldots, \lambda_{M,t})$ with constraints \eqref{Sceduling_constraint1} and a non-negative Lagrange multiplier $\mu_t$ with constraint \eqref{Sceduling_constraint2} at each time $t$.  
%we get the following problem: 
%\begin{align}\label{lagrangian}
%&\inf_{\pi \in \Pi}\sum_{t=0}^{\infty} \frac{\gamma^t}{K} \sum_{m=1}^M  \sum_{j=1}^{k_m} \mathbb{E}_{\pi} \left[ w_{m, j} p_{m, j}(\Delta_{m, j}(t))\right]\nonumber\\
%&~~~~~~+\sum_{t=0}^{\infty} \sum_{m=1}^M \lambda_{m, t}\frac{\gamma^t}{K} \left(\left(\sum_{j=1}^{k_m}   \pi_{m, j}(t)\right)-C_m\right)\nonumber\\
%&~~~~~~+\sum_{t=0}^{\infty} \mu_t\frac{\gamma^t}{K} \left(\left(\sum_{m=1}^M  \sum_{j=1}^{k_m} \pi_{m, j}(t)n_{m, j}\right)- N\right).
%\end{align}
%\chris{Could shorten by just presenting (8)?} 
To avoid an infinite number of Lagrange multipliers associated with the constraints over the infinite time horizon from $t=0$ to $t=\infty$, we truncate the problem to a finite time horizon $T$ as shown in \eqref{dual}.  Due to bounded inference error function, the performance loss resulting from this truncation becomes negligible for sufficiently large values of $T$.

{\violet The truncated problem from any time $\tau\in \{0, 1, \ldots, T\}$ to $T$ is given by
\begin{align}\label{dual}
&\bar p_{\tau}(\boldsymbol{\lambda}_{\tau:T}, \boldsymbol{\mu}_{\tau:T}) =\nonumber\\
&\;\inf_{\pi \in \Pi}\sum_{t=\tau}^{T}\frac{\gamma^{t-\tau}}{K}\sum_{m=1}^M  \sum_{j=1}^{k_m}   \mathbb{E}_{\pi}\left[  w_{m,j} p_{m, j}(\Delta_{m, j}(t))\right]\nonumber\\
&+\sum_{t=\tau}^{T} \sum_{m=1}^M {\lambda_{m, t}}\frac{\gamma^{t-\tau}}{K}  \left(\left(\sum_{j=1}^{k_m} \pi_{m, j}(t)\right)-C_m\right)\nonumber\\
&+\sum_{t=\tau}^{T} {\mu_t}\frac{\gamma^{t-\tau}}{K}  \left(\left(\sum_{m=1}^M  \sum_{j=1}^{k_m}\pi_{m, j}(t)n_{m, j}\right)- N\right),
\end{align}
where 
\begin{align}\label{langrange1}
    \boldsymbol{\lambda}_{t}=&(\lambda_{1,t}, \lambda_{2, t}, \ldots, \lambda_{M, t}),\\ \label{langrange2}
\boldsymbol{\lambda}_{\tau:T} =&(\boldsymbol{\lambda}_{\tau}, \boldsymbol{\lambda}_{\tau+1}, \ldots, \boldsymbol{\lambda}_T), \\ \label{langrange3}
\boldsymbol{\mu}_{\tau:T}=&(\mu_{\tau}, \mu_{\tau+1}, \ldots, \mu_T),
\end{align}
and $\bar p_{\tau}(\boldsymbol{\lambda}_{\tau:T}, \boldsymbol{\mu}_{\tau:T})$ is the optimal value of \eqref{dual}.
 
The dual problem to \eqref{dual} is given by 
\begin{align}\label{dualProblem}
(\boldsymbol{\lambda}^*_{\tau:T}, \boldsymbol{\mu}^*_{\tau:T}) = \argmax_{(\boldsymbol{\lambda}_{\tau:T}, \boldsymbol{\mu}_{\tau:T})\geq 0}~\bar p_{\tau}(\boldsymbol{\lambda}_{\tau:T}, \boldsymbol{\mu}_{\tau:T}),
\end{align}
where $(\boldsymbol{\lambda}^*_{\tau:T}, \boldsymbol{\mu}^*_{\tau:T})$ is the optimal dual solution.}

\subsection{Optimal Solution to \eqref{dual}}
{\violet The problem \eqref{dual} can be decomposed into $K$ sub-problems, one per task, in which the $(m, j)$-th sub-problem is given by 
\begin{align}\label{subproblem}
%&\bar p_{m, j}(\boldsymbol{\lambda}_m(\tau), \boldsymbol{\mu}(\tau); \tau:T)\nonumber\\
%=&
\inf_{\pi_{m,j} \in \Pi_{m, j}}\sum_{t=\tau}^{T} \gamma^{t-\tau} &\mathbb{E}_{\pi_{m,j}}  \bigg[w_{m, j} p_{m, j}(\Delta_{m, j}(t))\nonumber\\
&+ \lambda_{m, t} \pi_{m, j}(t)+ \mu_t \pi_{m, j}(t)n_{m, j}\bigg],
\end{align}
where %$\bar p_{m, j}(\boldsymbol{\lambda}_m(\tau), \boldsymbol{\mu}(\tau); \tau:T)$ is the optimal objective value of the sub-problem, $\boldsymbol{\lambda}_m(\tau)=(\lambda_{m, \tau}, \lambda_{m,\tau+1}, \ldots, \lambda_{m, T}),$ 
$\pi_{m, j}=(\pi_{m, j}(\tau), \ldots, \pi_{m, j}(T))$ is a scheduling policy for the task and $\Pi_{m, j}$ is the set of all causal signal-ignorant policies.} 
%The Lagrange multipliers $\boldsymbol{\lambda}_m(\tau)$ and $\boldsymbol{\mu}(\tau)$ correspond to the computation cost and the communication cost terms from time $\tau$ to $T$, respectively.}

By solving the sub-problem \eqref{subproblem} for each $(m, j)$-th MDP and combining the solutions, we get an optimal policy for \eqref{dual}. Following this approach, we present an optimal policy to the sub-problem \eqref{subproblem} in Lemma \ref{lemma1}.  

\begin{lemma}\label{lemma1}
There exists an optimal policy for \eqref{subproblem} in which the optimal decision $\pi^*_{m, j}(t)$ at each time $t$ minimizes the action value function 
\begin{align}\label{optimalsubpolicy}
\min_{\pi_{m,j}(t)\in \{0, 1\}} Q^{\boldsymbol{\lambda}_{m, t:T}, \boldsymbol{\mu}_{t:T}}_{m, j, t}( \Delta_{m, j}(t),\pi_{m, j}(t)), 
\end{align}
where the action value function $Q^{\boldsymbol{\lambda}_{m, t:T}, \boldsymbol{\mu}_{t:T}}_{m, j, t}( \cdot,\cdot)$ is given by
\begin{align}
&Q^{\boldsymbol{\lambda}_{m, t:T}, \boldsymbol{\mu}_{t:T}}_{m, j, t}( \delta,a)\nonumber\\
&= w_{m, j} p_{m,j}(\delta)+(1-a)\gamma V^{\boldsymbol{\lambda}_{m, t+1:T}, \boldsymbol{\mu}_{t+1:T}}_{m, j, t+1}(\delta+1)\nonumber\\
&\quad +a\left(\lambda_{m,t}+ \mu_t n_{m, j}+ \gamma V^{\boldsymbol{\lambda}_{m, t+1:T}, \boldsymbol{\mu}_{t+1:T}}_{m, j, t+1}(1)\right),
\end{align}
the value function $V^{\boldsymbol{\lambda}_{m, t:T}, \boldsymbol{\mu}_{t+1:T}}_{m, j, t}(\delta)$ for all $\delta \in \mathbb Z^{+}$ and $t=\tau,\tau+1, \ldots, T$ is given by 
\begin{align}\label{valuefunction}
V^{\boldsymbol{\lambda}_{m, t:T}, \boldsymbol{\mu}_{t:T}}_{m, j, t}(\delta)=\min_{a\in \{0, 1\}}Q^{\boldsymbol{\lambda}_{m, t:T}, \boldsymbol{\mu}_{t:T}}_{m, j, t}( \delta,a), 
\end{align}
and for $t=T+1,$
\begin{align}
V^{\boldsymbol{\lambda}_{m, T+1:T}, \boldsymbol{\mu}_{T+1:T}}_{m, j, T+1}(\delta)=0.
\end{align}
\end{lemma}
\begin{proof}
The action $\pi^*_{m,j}(t)$ is optimal because it satisfies the Bellman optimality equation \cite{bertsekasdynamic1, puterman2014markov}: \begin{align}\label{Optimality}
&V^{\boldsymbol{\lambda}_{m, t:T}, \boldsymbol{\mu}_{t:T}}_{m, j, t}(\delta)\nonumber\\
&\!\!=\min_{a \in \{0,1\}}w_{m, j}p_{m, j}(\delta)+a\left(\lambda_{m,t}+\mu_{t} n_{m,j}\right)\nonumber\\
&+\gamma \mathbb E\!\left[V^{\boldsymbol{\lambda}_{m, t+1, T}, \boldsymbol{\mu}_{t+1:T}}_{m, j, t+1}(\Delta_{m,j}(t+1))\bigg|\Delta_{m,j}(t)\!=\!\delta,\! \pi_{m,j}(t)\!=a\!\right].
\end{align}
\end{proof}

Lemma \ref{lemma1} establishes an optimal decision $\pi^*_{m, j}(t)$ for problem \eqref{subproblem} by using dynamic programming method. The backward induction method to compute the value function $V^{\boldsymbol{\lambda}_{m, t:T}, \boldsymbol{\mu}_{t:T}}_{m, j, t}(\delta)$ for $t=\tau, \tau+1, \ldots, T$ is given by
\begin{align}
&V^{\boldsymbol{\lambda}_{m, t:T}, \boldsymbol{\mu}_{t:T}}_{m, j, t}(\delta)\nonumber\\
&=\!w_{m, j} p_{m,j}(\delta)+\!\min_{a \in \{0, 1\}}\bigg\{\!(1-a)\gamma V^{\boldsymbol{\lambda}_{m, t+1:T}, \boldsymbol{\mu}_{t+1:T}}_{m, j, t+1}(\delta+1)\nonumber\\
&+a\left(\lambda_{m,t}+\mu_t n_{m,j}+\gamma V^{\boldsymbol{\lambda}_{m, t+1:T}, \boldsymbol{\mu}_{t+1:T}}_{m, j, t+1}(1)\right)\bigg\}.
\end{align}
However, if the AoI $\delta$ can take infinite values, this is computationally intractable. Thus, we restrict the computation of value function to a finite range $\delta=1, 2, \ldots, \bar \delta$, and approximate $V^{\boldsymbol{\lambda}_{m, t:T}, \boldsymbol{\mu}_{t:T}}_{m, j, t}(\delta) \approx V^{\boldsymbol{\lambda}_{m, t:T}, \boldsymbol{\mu}_{t:T}}_{m, j, t}(\bar \delta)$ for values exceeding this range. In reality, this truncation will have a negligible effect since (i) higher AoI values are rarely visited in practice \cite{Ceran}, and (ii) the inference error $p_{m, j}(\delta)$ tends to converge to an upper bound as AoI becomes large, as seen in some recent works \cite{shishertimely, ShisherMobihoc, shisher2023learning} and our machine learning experiments in Figs. \ref{fig:roboterror} \& \ref{fig:traffic}. The backward induction algorithm has a time complexity of $O(\bar\delta T)$.

\subsection{Solution to \eqref{dualProblem}}

% \begin{lemma}\label{concavity}
%     The optimal objective value $\bar p(\boldsymbol{\lambda}(\tau), \boldsymbol{\mu}(\tau); \tau:T)$ of the problem \eqref{dual} is concave in $\boldsymbol{\lambda}(\tau)$ and  $\boldsymbol{\mu}(\tau)$.
% \end{lemma}

{\violet Next, we solve the dual problem \eqref{dualProblem}. Similar to \cite[Proposition 3.1(c)]{brown2023fluid}, we can show that the optimal objective value $\bar p_{\tau}(\boldsymbol{\lambda}_{\tau:T}, \boldsymbol{\mu}_{\tau:T})$ of the problem \eqref{dual} is concave in $\boldsymbol{\lambda}_{\tau:T}$ and  $\boldsymbol{\mu}_{\tau:T}$. Because of the concavity, we can solve the dual problem \eqref{dualProblem} by the stochastic sub-gradient ascent method. The $(i+1)$-th iteration for the sub-gradient ascent method for all $m=1, 2, \ldots$ and $t=\tau, \tau+1, \ldots, T$ is given by

\begin{align}\label{lambdaopt}
    &\lambda_{m, t}(i+1)=\nonumber\\
    &\max\bigg\{\lambda_{m, t}(i)+ \frac{\gamma^{t-\tau}\beta_m}{Ki}\left(\sum_{j=1}^{k_m}\pi_{m, j}^*(t)-C_m\right), 0\bigg\}, \\
     &\mu_{t}(i+1)=\nonumber\\\label{muopt}
    &\max\bigg\{\mu_{t}(i)+ \frac{\gamma^{t-\tau}\beta}{Ki}\left(\sum_{m=1}^M\sum_{j=1}^{k_m}\pi_{m, j}^*(t)n_{m,j}-N\right), 0\bigg\},
\end{align}
where $\beta, \beta_m>0$ are the step sizes and $\pi_{m, j}^*(t)$ is the optimal action to the $(m, j)$-th sub-problem determined with $\lambda_{m, t}(i), \lambda_{m, t+1}(i), \ldots, \lambda_{m, T}(i)$, and $\mu_{t}(i), \mu_{t+1}(i), \ldots, \mu_{T}(i)$ obtained in the $i$-th iteration.} 

\section{Scheduling Policy}
\label{sec:policy}

\subsection{Reoptimized Maximum Gain First (MGF) Policy}
\label{ssec:mgf}
While the decision $\pi^*_{m,j}(t)$ provided in Lemma \ref{lemma1} may violate constraints \eqref{Sceduling_constraint1}-\eqref{Sceduling_constraint2}, we exploit the structure of the decision $\pi^*_{m,j}(t)$ to develop a scheduling policy for the original problem \eqref{Multi-scheduling_problem}-\eqref{Sceduling_constraint2}. The proposed policy utilizes the notion of gain indices discussed in some recent papers \cite{shishertimely, OrneeMILCOM, chen2023index}. %To determine gain indices for our MTRI problem, we use the optimal solutions $(\boldsymbol{\lambda}^*_{m}, \boldsymbol{\mu}^*)$ to the dual problem \eqref{dualProblem} and the action value function $Q^{\boldsymbol{\lambda}^*_m, \boldsymbol{\mu}^*}_{m, j, t}(\cdot)$.
{\blue To determine gain indices for our MTRI problem at time $t$, we use the action value function $Q^{\boldsymbol{\lambda}_{m, t:T}^*, \boldsymbol{\mu}_{t:T}^*}_{m, j, t}( \Delta_{m, j}(t),\pi_{m, j}(t))(\cdot)$ associated with Lagrange multipliers $\boldsymbol{\lambda}_{m, t:T}^*$ and $\boldsymbol{\mu}_{t:T}^*$.}

\begin{definition}[\textbf{Gain Index}]\label{gainindex}\cite{shishertimely}
{\blue Given an AoI value $\Delta_{m, j}(t)=\delta$, the gain index $\alpha_{m, j, t}(\delta)$ for the $(m,j)$-th task at time $t$ is the difference of two actions values, determined by 
\begin{align}\label{gain}
\alpha_{m, j, t}(\delta)=&Q^{\boldsymbol{\lambda}_{m, t:T}^*, \boldsymbol{\mu}_{t:T}^*}_{m, j, t}( \Delta_{m, j}(t),\pi_{m, j}(t))(\delta,0)\nonumber\\
&-Q^{\boldsymbol{\lambda}_{m, t:T}^*, \boldsymbol{\mu}_{t:T}^*}_{m, j, t}( \Delta_{m, j}(t),\pi_{m, j}(t))(\delta,1),
\end{align}
where the Lagrange multipliers
are obtained after solving \eqref{dualProblem} with $\tau=t$.}
\end{definition}

\begin{algorithm}[t]
\SetAlgoLined
\SetAlgoNoEnd
\SetKwInOut{Input}{input}
\caption{\small Reoptimized Maximum Gain First Policy}\label{alg:gain}
%\begin{algorithmic}[1]
%\Input{\emph{Gain indices $\alpha_{m, j, t} (\delta)$ for all $t=0, 1, \ldots, T$ and $\delta=1, 2, \ldots$ using \eqref{gain} and gain indices $\alpha_{m, j, t} (\delta)=\alpha_{m, j, T} (\delta)$} for all $t=T+1, T+2, \ldots$ and $\delta=1, 2, \ldots.$}
%\BlankLine
\For{$t=0, 1, \ldots$}
{\emph{Update $\Delta_{m, j}(t)$ for all $(m, j)$}\\
Initialize $\pi_{m, j}(t) \gets 0 $ for all $(m, j)$\\
{\blue \emph{Get $\boldsymbol{\lambda}^*$ and $\mu^*$ that maximizes $\bar p(\boldsymbol{\lambda}(t), \boldsymbol{\mu}(t); t:T)$}}\\
\emph{$\alpha_{m, j} \gets \alpha_{m, j, t} (\Delta_{m, j}(t))$} for all $(m, j)$\\

$C_{m, \text{curr}} \gets 0$ and $N_{\text{curr}} \gets 0$\\
 $A(t)\gets \{(m, j):\alpha_{m, j}>0)\}$\\
%\State Feasible $\gets$ True and $K \gets 0$.
\While{$A(t)$ is not empty}
{{\blue $(m^*, j^*) \gets \argmax_{m, j} \alpha_{m, j}$}\\
$c \gets C_{m^*, \text{curr}}+1$ and $n \gets N_{\text{curr}}+n_{m^*,j^*}$\\
\If{$c \leq C_{m^*}$ and $n \leq N$}
{Update $\pi_{m^*,j^*}(t)\gets 1$ \\
Update $C_{m^*, \text{curr}} \gets c$ and $N_{\text{curr}} \gets n$}
$A(t)=A(t)\setminus(m^*, j^*)$}}
\end{algorithm}

The gain index $\alpha_{m, j, t}(\delta)$ quantifies the discounted total reduction in inference errors when action $\pi_{m, j}(t)=1$ is chosen over $\pi_{m, j}(t)=0$, where the latter implies no resource allocation for the $(m,j)$-th inference task at time $t$. This metric enables strategic resource utilization at each time slot to enhance overall system performance. %As the gain index is computable only up to a finite time $T$, for $t > T$, we assume $\alpha_{m, j, t}(\delta) = \alpha_{m, j, T}(\delta)$.

{\violet Algorithm \ref{alg:gain} presents our reoptimized maximum gain first (MGF) scheduler for solving the main problem \eqref{Multi-scheduling_problem}-\eqref{Sceduling_constraint2}. We denote $\pi_{\text{MGF}}$ by the policy provided in Algorithm \ref{alg:gain}. At each time $t$, the policy $\pi_{\text{MGF}}$ prioritizes generating and transmitting features ($\pi_{m, j}(t)=1$) for the inference tasks with highest gain index, while adhering to the available communication and computation resources. Our policy then proceeds as follows: 
\begin{itemize}
\item[(1)] First, our policy re-optimizes $\boldsymbol{\lambda}^*(t)$ and $\mu^*(t)$ by maximizing $\bar p(\boldsymbol{\lambda}(t), \boldsymbol{\mu}(t); t:T)$. Then, the gain indices $\alpha_{m, j, t}(\Delta_{m, j}(t))$ for all tasks $(m, j)$ defined in \eqref{gain} are calculated.   

\item[(2)] Let $A(t)$ be the set of inference tasks with positive gain indices:
\begin{align}
    A(t)=\bigg\{(m, j): \alpha_{m, j, t}(\Delta_{m, j}(t))>0\bigg\}
\end{align} 
Then, our policy selects the inference task $(m^*, j^*)$ that satisfies
\begin{align}\label{gainimplement}
    (m^*, j^*)=\argmin_{(m, j)\in A(t)}\alpha_{m, j, t}(\Delta_{m, j}(t)).
\end{align}
Source $m^*$ generates and transmits its features for the $(m^*, j^*)$-th inference task, provided that the resource budget has not been exhausted. 
\item[(3)] Remove the tuple $(m^*, j^*)$ from $A(t)$, i.e., $A(t)=A(t)\setminus(m^*, j^*)$. Repeat (1) until the set $A(t)$ is empty.
\end{itemize}
Comparing \eqref{optimalsubpolicy}, \eqref{gain}, and \eqref{gainimplement}, we observe that our policy closely approximates the optimal solution to the Lagrangian relaxed problem \eqref{dual}, aiming to make as close to full use of the resource constraints as possible.}

\subsection{Performance Analysis}
\label{ssec:analysis}
We now analyze the performance of our policy relative to the original problem \eqref{Multi-scheduling_problem}-\eqref{Sceduling_constraint2}. Following standard practice in the weakly-coupled MDP literature \cite{brown2023fluid, gast2024reoptimization}, {\violet a set of sub-problems at source $m$ are said to be in the same class if they share identical penalty functions, weights, and transition probabilities.} 
%where each sub-problem is indexed by the $(m, j)$-th inference task.  

\begin{definition}[\textbf{Asymptotic optimality}]\label{defAsymptotically optimal}
{\violet Consider a ``base'' MTRI system with $N$ channels, $M$ sources, $k_m$ classes of sub-problems per source $m$, a computation resource budget $C_m$ for source $m$ and $N$ shared communication resources by $M$ sources. Now, consider a ``multiplied system", where each class of sub-problems per source $m$ contains $r$ inference tasks, each source $m$ has $rC_m$ computation resources and $rN$ communication resources, while maintaining a constant $M$ sources. Let $\bar p^r_{{\text{MGF}}}$ and $\bar p^{r}_{{\pi}_{opt}}$ represent the discounted infinite horizon sum of inference errors under policy $\pi_{\text{MGF}}$ and an optimal policy for the multiplied system, respectively. The policy $\pi_{\text{MGF}}$ is asymptotically optimal if $\bar p^r_{{\text{MGF}}}=\bar p^r_{{\pi}_{opt}}$ for all $\pi \in \Pi$ as inference task per class $r$ approaches $\infty$, i.e.,
\violet \begin{align}
    \lim_{r\to \infty} \bar p^r_{{\text{MGF}}}=\bar p^r_{{\pi}_{opt}}.
\end{align}}
\end{definition}

{\violet First, we establish Lemma \ref{lemmatheorem4}, which is a key tool to showing asymptotic optimality. We define a policy $\pi^*=(\pi^*_{m,j}(t))_{\forall m,j, t\leq T}$, where $\pi^*_{m, j}(t)$ is the reoptimized action obtained by using Lemma \ref{lemma1} with the Lagrange multipliers that maximize $\bar p_t(\boldsymbol{\lambda}_{t:T}, \boldsymbol{\mu}_{t:T})$ defined in \eqref{dualProblem}. Using \eqref{optimalsubpolicy} and \eqref{gain}, we can verify that $\pi^*_{m, j}(t) = 1$ if $\alpha_{m,j,t}(\Delta_{m,j}(t))>0$.

\begin{lemma}\label{lemmatheorem4}
   For any time $t$ and AoI values, the expected number of subproblems with actions that differ between the reoptimized MGF policy provided in Algorithm \ref{alg:gain} and the policy $\pi^*$ is bounded from above by
   $$\sum_{m=1}^M \sqrt{k_m}+\sqrt{\sum_{m=1}^M k_m}.$$
\end{lemma}  
\begin{proof}
    See Appendix \ref{plemmatheorem4}.
\end{proof}}

Then, we can obtain our main theoretical result:
 
{\violet \begin{theorem}\label{theorem4}
If 
%the functions $w_{m, j} p_{m, j}(\delta)$ are bounded for all tasks $(m, j)$, i.e., 
there exists a finite constants $\bar p_l$ and $\bar p_h$ such that $\bar p_l \leq w_{m, j} p_{m, j}(\delta) \leq \bar p_h$ for any sub-problem $(m, j)$ and AoI value $\delta=1, 2, \ldots,$ and %{\color{red} Yin: What is $\alpha$?}
\begin{align}
T \geq \log_{\frac{1}{\gamma}}\left(\sum_{m=1}^M \sqrt{rk_m}\right),
\end{align}
then the MGF policy is asymptotically optimal as the number of inference tasks $r$ for all classes $(m,j)$ increases to infinite. Specifically, we have 
\begin{align}\label{result}
\bar p^r_{{\text{MGF}}}-\bar p^r_{\text{opt}}&\leq \frac{1}{\sum_{m=1}^M \sqrt{rk_m}} 
\left(\frac{2M(\bar p_h-\bar p_l)\gamma}{(1-\gamma)^3} + \frac{(\bar p_h-\bar p_l)\gamma}{(1-\gamma)}\right)\nonumber\\
&=O\bigg(\frac{1}{\sum_{m=1}^M\sqrt{rk_m}}\bigg),
\end{align}
where $k_m$ is the number of sub-problems per source, $M$ is the number of sources, and $\gamma$ is the discount factor.
\begin{proof}
    See Appendix \ref{ptheorem4}.   
\end{proof}
\end{theorem}}

%{\color{red} Yin: Check the division of all terms in Equation (4) of [13] by $N$. Check Section 6 of [13].  Check [14]. For Theorem 2 to be correct, you need to scale up the amounts of computation resources and channel resources linearly by $k$. Following Section 5.2 of [13], you should define a fixed number of $k$ task types. Note that $k$ should be the total number of task types across all sources. Some class types consume certain computational resources at one source, some class types consume other computational resources at another source; the association of task type to sources are denoted by $l$ function in (4) and (15) of [13]. }

According to Theorem \ref{theorem4} and Definition \ref{defAsymptotically optimal}, our policy approaches the optimal as the number of inference tasks $r$ per class of sub-problems increases asymptotically. 

While prior work has introduced gain-index-based policies for RMAB problems \cite{OrneeMILCOM, chen2023index, shisher2023learning}, these cannot be directly applied to general weakly-coupled MDPs. {\blue Our gain-index-based policy, a specialized re-optimized fluid policy (Definition \ref{ROFP}) \cite{brown2023fluid}, achieves tighter asymptotic optimality for MTRI systems (see \eqref{result}) compared to the $O(\frac{1}{\sqrt{\sum_{m=1}^Mrk_m}})$ bound in \cite{brown2023fluid}. This improvement is obtained by using Lemma \ref{lemmatheorem4}, which strengthens the result \cite[Lemma EC1.1]{brown2023fluid} by exploiting the MTRI constraint structure: sub-problems utilize only their source's computational resources. Unlike the general system in \cite{brown2023fluid} where all resources are globally shared, MTRI systems have local (computational) and global (communication) resources, yielding a tighter bound. For example, a sub-problem associated with source $m_1$ would only consume computational resources from source $m_1$ and not from another source $m_2$, unlike in \cite{brown2023fluid}, where all sub-problems share all resources in the system.}

\begin{algorithm}[t]
\SetAlgoLined
\SetAlgoNoEnd
\SetKwInOut{Input}{input}
\caption{Simplified Reoptimized MGF Policy }\label{alg:gain1}
%\begin{algorithmic}[1]
%\Input{\emph{Gain indices $\alpha_{m, j, t} (\delta)$ for all $t=0, 1, \ldots, T$ and $\delta=1, 2, \ldots$ using \eqref{gain} and gain indices $\alpha_{m, j, t} (\delta)=\alpha_{m, j, T} (\delta)$} for all $t=T+1, T+2, \ldots$ and $\delta=1, 2, \ldots.$}
%\BlankLine
\For{$t=0, 1, \ldots$}
{\emph{Update $\Delta_{m, j}(t)$ for all $(m, j)$}\\
Initialize $\pi_{m, j}(t) \gets 0 $ for all $(m, j)$\\
{\violet \emph{Get $\lambda_1^*, \lambda^*_2, \ldots, \lambda^*_M$, and $\mu^*$ that maximizes $\bar p_{t}(\lambda_1,\lambda_2, \ldots, \lambda_M, \mu)$}}\\
\emph{$\alpha_{m, j} \gets \alpha_{m, j, t} (\Delta_{m, j}(t))$} for all $(m, j)$\\

$C_{m, \text{curr}} \gets 0$ and $N_{\text{curr}} \gets 0$\\
 $A(t)\gets \{(m, j):\alpha_{m, j}>0)\}$\\
%\State Feasible $\gets$ True and $K \gets 0$.
\While{$A(t)$ is not empty}
{{\blue $(m^*, j^*) \gets \argmax_{m, j} \alpha_{m, j}$}\\
$c \gets C_{m^*, \text{curr}}+1$ and $n \gets N_{\text{curr}}+n_{m^*,j^*}$\\
\If{$c \leq C_{m^*}$ and $n \leq N$}
{Update $\pi_{m^*,j^*}(t)\gets 1$ \\
Update $C_{m^*, \text{curr}} \gets c$ and $N_{\text{curr}} \gets n$}
$A(t)=A(t)\setminus(m^*, j^*)$}}
\end{algorithm}

{\violet \subsection{Simplified Reoptimized Maximum Gain First Policy }\label{modification} In Algorithm \ref{alg:gain}, at every time $t$, we require to optimize the Lagrange variables
$\boldsymbol{\lambda}_{t:T}$ and $\boldsymbol{\mu}_{t:T}$. The definition \eqref{langrange1}-\eqref{langrange3} show that the total number of parameters to optimize is $(M+1)(T-t)$ at each time $t$. The computational complexity of such optimization may not be feasible for real-time implementation. To handle this issue, we keep the Lagrange variables $\lambda_{m, t}=\lambda_m$ and $\mu_t=\mu$ for all $t=1,2,\ldots, T$. This reduces the number of optimization variables to $(M+1)$. After adopting time-invariant Lagrange variables, the primal problem \eqref{dual} becomes:
\begin{align}\label{dual1}
&\bar p_{\tau}(\lambda_1,\lambda_2, \ldots, \lambda_M, \mu) =\nonumber\\
&\;\inf_{\pi \in \Pi}\sum_{t=\tau}^{T}\sum_{m=1}^M  \sum_{j=1}^{k_m}  \frac{\gamma^{t-\tau} \mathbb{E}_{\pi}\left[  w_{m,j} p_{m, j}(\Delta_{m, j}(t))\right]}{K}\nonumber\\
&+ \sum_{m=1}^M \lambda_{m} \sum_{t=\tau}^{T}\frac{\gamma^{t-\tau}}{K}  \left(\left(\sum_{j=1}^{k_m} \pi_{m, j}(t)\right)-C_m\right)\nonumber\\
&+{\mu} \sum_{t=\tau}^{T} \frac{\gamma^{t-\tau}}{K}  \left(\left(\sum_{m=1}^M  \sum_{j=1}^{k_m}\pi_{m, j}(t)n_{m, j}\right)- N\right),
\end{align}
and the dual problem \eqref{dualProblem} becomes
\begin{align}\label{dualProblem1}
(\lambda_1^*,\lambda^*_2, \ldots, \lambda^*_M, \mu^*) = \argmax_{\lambda_1,\lambda_2, \ldots, \lambda_M, \mu\geq 0} \bar p(\boldsymbol{\lambda}(\tau), \boldsymbol{\mu}(\tau); \tau:T).
\end{align}

The dual sub-gradient ascent algorithm for solving \eqref{dualProblem1} is as follows:

\begin{align}
    &\lambda_{m}(i+1)=\max\bigg\{\lambda_{m}(i)+\nonumber\\
    &\frac{\beta_m}{Ki}\left(\sum_{t=\tau}^T\sum_{j=1}^{k_m}\gamma^{t-\tau}\pi_{m, j}^*(t)-\frac{(1-\gamma^{(T-t)})C_m}{(1-\gamma)}\right), 0\bigg\}, \\
     &\mu(i+1)=\max\bigg\{\mu(i)+ \nonumber\\
     &\frac{\beta}{Ki}\left(\sum_{t=\tau}^T\sum_{m=1}^M\sum_{j=1}^{k_m}\gamma^{t-\tau}\pi_{m, j}^*(t)n_{m,j}-\frac{(1-\gamma^{(T-t)})N}{(1-\gamma)}\right), 0\bigg\},
\end{align}
where $\beta, \beta_m>0$ are the step sizes and $\pi_{m, j}^*(t)$ is the optimal action to the $(m, j)$-th sub-problem with $\lambda_{m}(i)$ and $\mu(i)$ obtained in the $i$-th iteration.}

\section{Numerical Experiments}
\label{sec:simulation}

\begin{figure}[t]
  \centering
\includegraphics[width=0.4\textwidth]{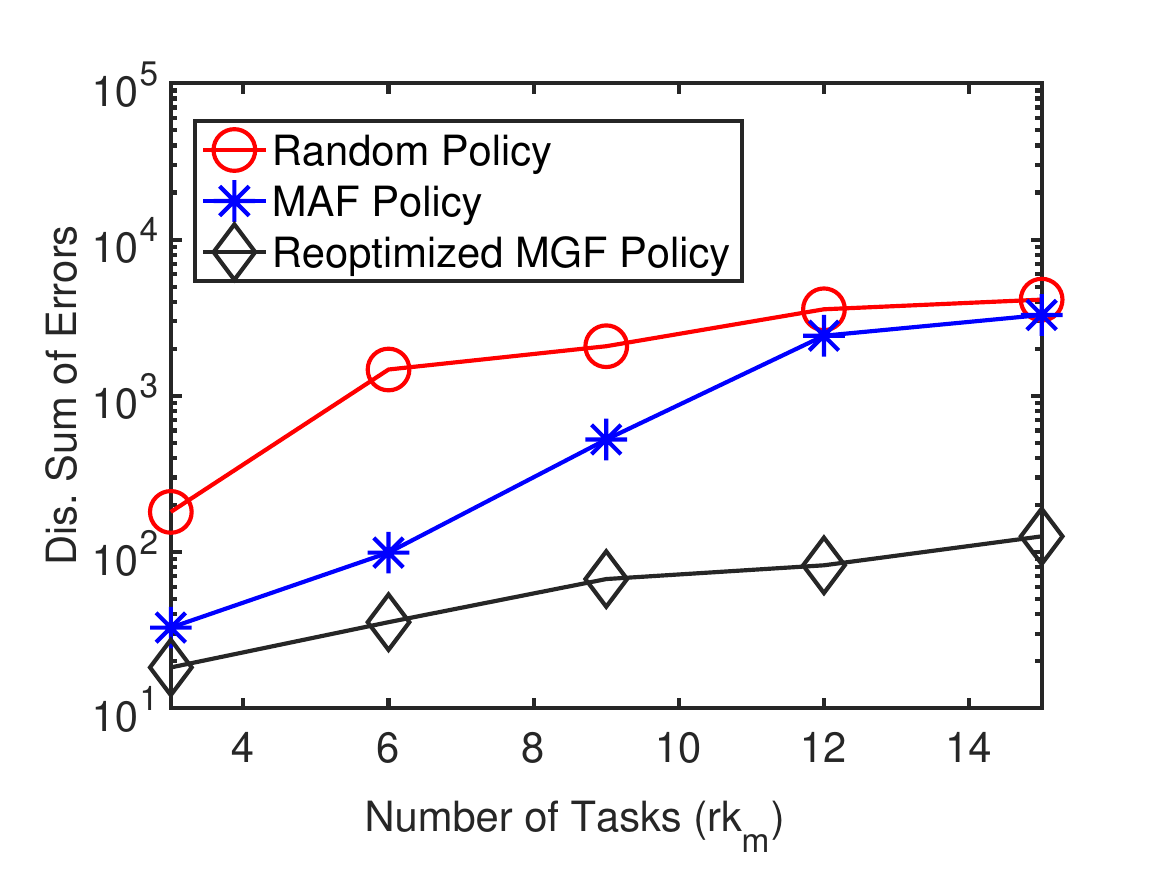}
  \caption{Dis. Sum of Errors vs. $rk_m$}\label{fig:synthetictask}
\end{figure}

\begin{figure}[t] \centering
\includegraphics[width=0.4\textwidth]{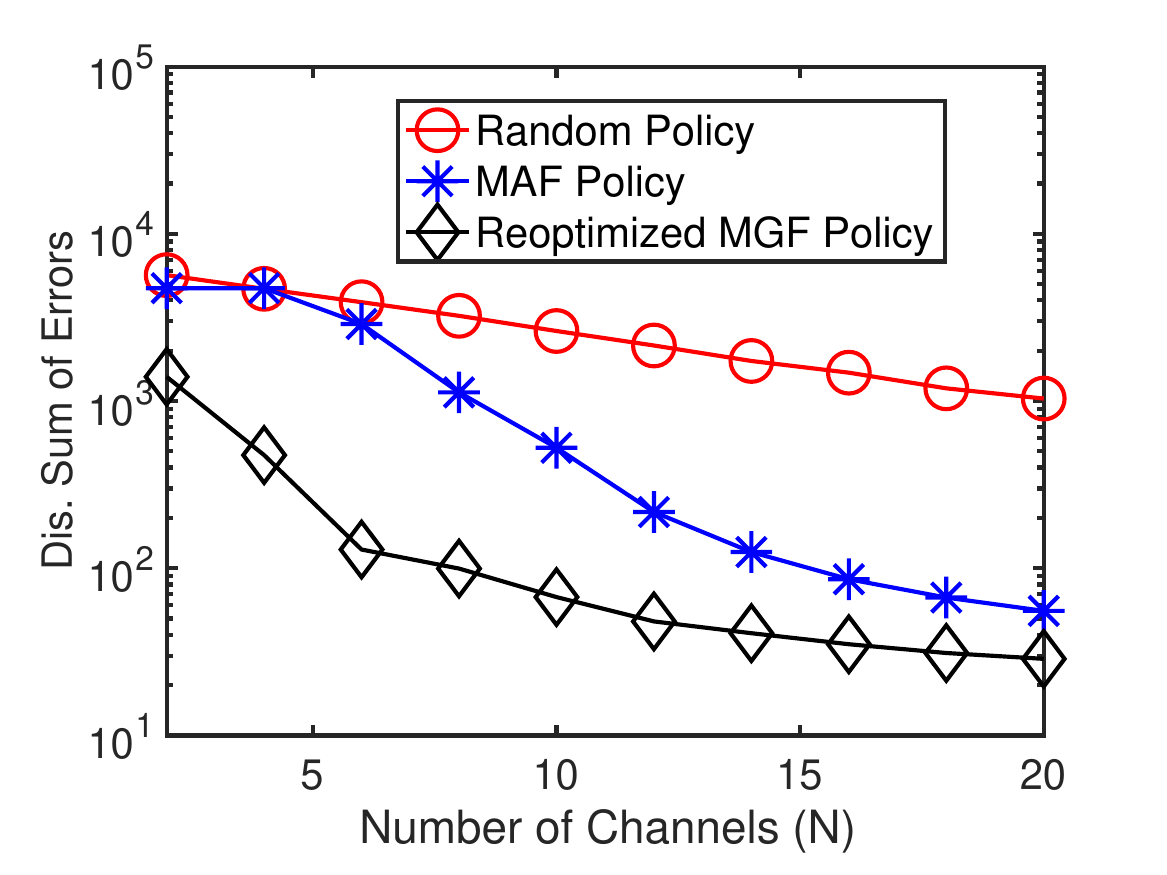}
\caption{Dis. Sum of Errors vs. $N$}\label{fig:syntheticchannel}
\end{figure}

\begin{figure}[t]\centering
\includegraphics[width=0.4\textwidth]{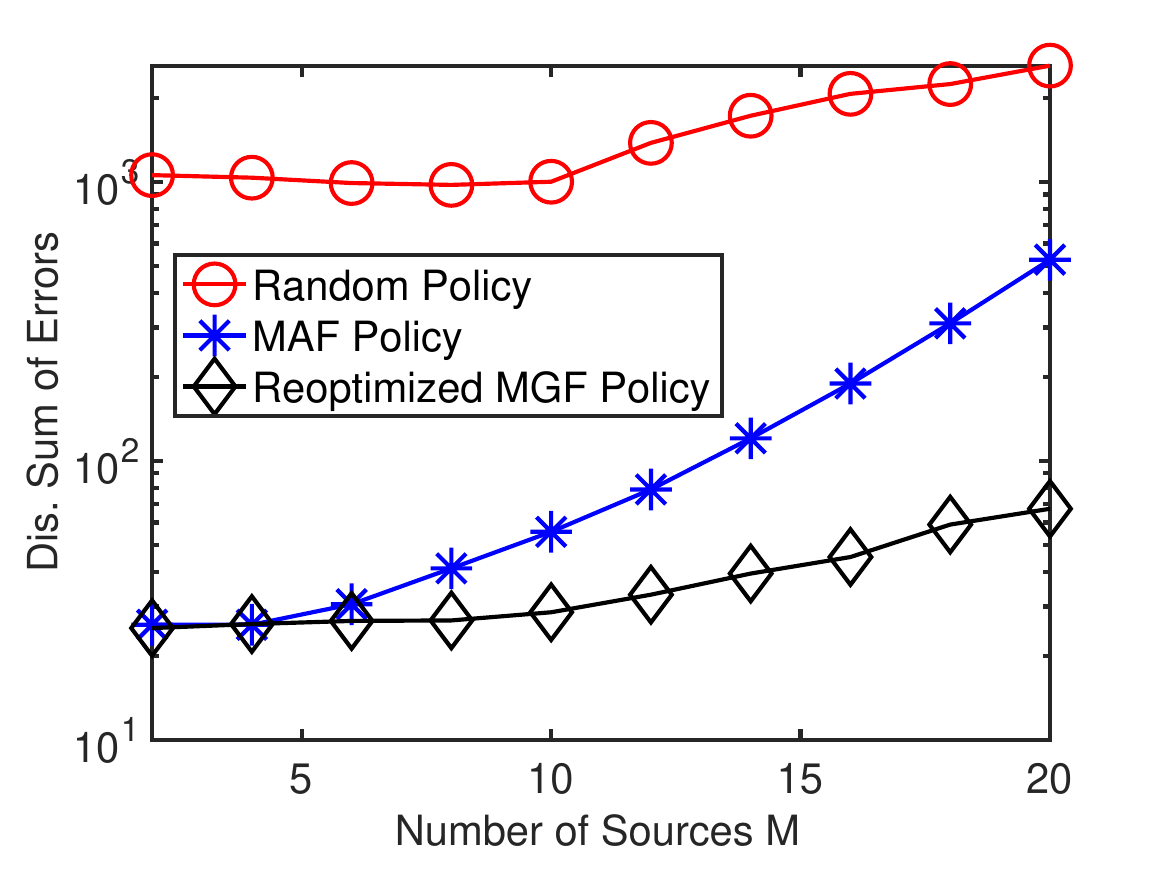}
\caption{Dis. Sum of Errors vs. $M$}\label{fig:syntheticsources}
\end{figure}

% \begin{figure*}[h]
%   \centering
  
%   \begin{subfigure}[t]{0.30\textwidth}
% \includegraphics[width=\textwidth]{ErrorVsTasks.eps}
%   \subcaption{Dis. Sum of Errors vs. $rk_m$}
% \end{subfigure}
% % <— this is important. There should be no empty line here. 
% %
% \hspace{1mm} 
% \begin{subfigure}[t]{0.30\textwidth}
% \includegraphics[width=\textwidth]{ErrorsVsChannelsmodel.eps}
% \subcaption{\small Dis. Sum of Errors vs. $N$}
% \end{subfigure}
% % <— this is important. There should be no empty line here. 
% %
% \hspace{1mm} 
% \begin{subfigure}[t]{0.30\textwidth}
% \includegraphics[width=\textwidth]{ErrorVsSourcesModel.eps}
% \subcaption{Dis. Sum of Errors vs. $M$}
% \end{subfigure}
% \vspace{-0.05in}
% \caption{\small Model-based evaluations: (a) discounted sum of errors vs. number of tasks per source $rk_m$, (b) discounted sum of errors vs. number of channels $N$, and (c) discounted sum of errors vs. number of sources $M$. The MGF policy has the potential to achieve a discounted sum of inference errors up to 26 and 32 times lower than the MAF and Random policies, respectively. \label{fig:model}}
% \vspace{-0.2in}
% \end{figure*}

We consider the following three policies for evaluation:
\begin{itemize}[leftmargin=7mm]
\item \textit{Maximum Gain First (MGF) Policy:} {\violet For numerical study, we use the modified Algorithm \ref{alg:gain1}.} 

\item \textit{Maximum Age First (MAF) Policy:} At each time slot $t$, the MAF policy selects the inference task $(m, j)$ with the highest AoI from the set of all inference tasks $A_1(t)$ with non-zero AoI. If constraints permit, the policy generates and transmits the feature for the selected task. Then, $(m, j)$ is removed from $A_1(t)$. This process repeats until $A_1(t)$ is empty. AoI-based priority policies are commonly used as baselines in the literature \cite{zou2021minimizing, shishertimely, sun2023age}.

\item \textit{Random Policy:} At each time slot $t$, the random policy selects one inference task $(m, j)$ from the set of all tasks $A_2(t)$ following a uniform distribution. If constraints permit, the policy generates and transmits the feature for the selected task. The task $(m, j)$ is then removed from $A_2(t)$. This process repeats until $A_2(t)$ becomes empty.
\end{itemize}

We evaluate these three policies under three scenarios: 
\begin{itemize}[leftmargin=7mm]
    \item \textit{Synthetic evaluations:} We assess the policies assuming synthetic AoI penalty functions for all inference tasks (Sec. \ref{Synthetic}).
    \item \textit{Remote Robot Car Detection:} {\violet We conduct a
    remote robot car detection experiment. We generate the inference error function for robot car detection and incorporate the resulting inference error functions into the simulation} (Sec. \ref{robotcar})
    \item \textit{Traffic Prediction and Segmentation:} In this experiment, we consider two machine learning tasks: (i) scene segmentation and (ii) traffic prediction on the NGSIM dataset \cite{NGSIM_I80_2016, NGSIM_US101_2016, NGSIM_Lankershim_2016, NGSIM_Peachtree_2016}. Then, we incorporate the resulting inference error functions into the simulation (Sec. \ref{machinelearning}).
\end{itemize}

\subsection{Synthetic Evaluations}\label{Synthetic} 
In this section, we use three AoI penalty functions: $p_{m,j}(\delta) = \delta, \exp(0.5 \delta), 10 \text{log} (\delta)$. These functions are widely used in AoI literature as estimation error \cite{Tripathi2019, SunNonlinear2019}. Each function is assigned to one-third of the inference tasks in each source $m$.

Fig. \ref{fig:synthetictask} illustrates the discounted sum of inference errors versus the number of tasks per source ($rk_m$) over a time horizon of $T=100$. Referring to Definition~\ref{defAsymptotically optimal}, we set $k_m=3$ and vary $r$. The additional simulation parameters are $M=20$, $N=10$, $\gamma=0.9$, $n_{m,j}=1$ for all tasks $(m, j)$, $C_m=2$ for all sources $m$, and $w_{m,j} = 0.01$ for half of the tasks and $1$ for the other half. We see that, when $rk_m=15$, the total discounted penalty for the MAF policy is $26$x higher than that of the MGF policy, while the random policy incurs $32$x higher penalty. 
%Although the MGF policy consistently outperforms the others, the MAF policy demonstrates good performance when $k_m=3$ or $6$.
The performance of the MAF policy deteriorates more rapidly than the MGF policy as the number of tasks per source increases, aligning with our findings in Theorem \ref{theorem4}.

In Fig. \ref{fig:syntheticchannel}, we plot the discounted sum of errors against the number of channels $N$. Here, $k_m = 3$ and $r=3$ for each source $m$, and the rest of the parameters are the same as in Fig. \ref{fig:syntheticchannel}.
%a time horizon of $T=100$. The parameters include $M=20$, $k_m=9$ for each source $m$, $\gamma=0.9$, $n_{m,j}=1$ for all $(m, j)$, $C_m=2$ for all $m$ with $w_{m,j}$ set to 0.01 for half the tasks and 1 for the rest.
We see that increasing $N$ improves performance for all policies, but more rapidly for MGF. When $N=2$, the MAF policy incurs four times penalty of the MGF policy. This performance gap narrows as $N$ increases, but even with $N=20$, the MAF policy's inference error remains twice that of the MGF policy.
%Also, it is clear from Fig. \ref{fig:model}(b) that random policy is worse that other two policies.

In Fig. \ref{fig:syntheticsources}, we plot the discounted sum of errors against the number of sources $M$, with other parameters the same as in Fig. \ref{fig:synthetictask}\&\ref{fig:syntheticchannel}. As the number of sources increases, we see that the performance gap between MAF and MGF policies widens. This shows that MGF is more effective as the number of sources competing for MTRI resources increases.
%$M = 10$. Fig. \ref{fig:model}(c) demonstrates the effectiveness of our proposed policy in system with numerous edge devices.
% over a finite time horizon $T=100$, with simulation parameters $N=10$, $k_m=9$ for each source $m$, $\gamma=0.9$, $n_{m,j}=1$, $C_m=2$ for all $m$, and $w_{m,j}$ set to 0.01 for half the tasks and 1 for the rest.

  \begin{figure}[t]
    \centering
\includegraphics[width=0.4\textwidth]{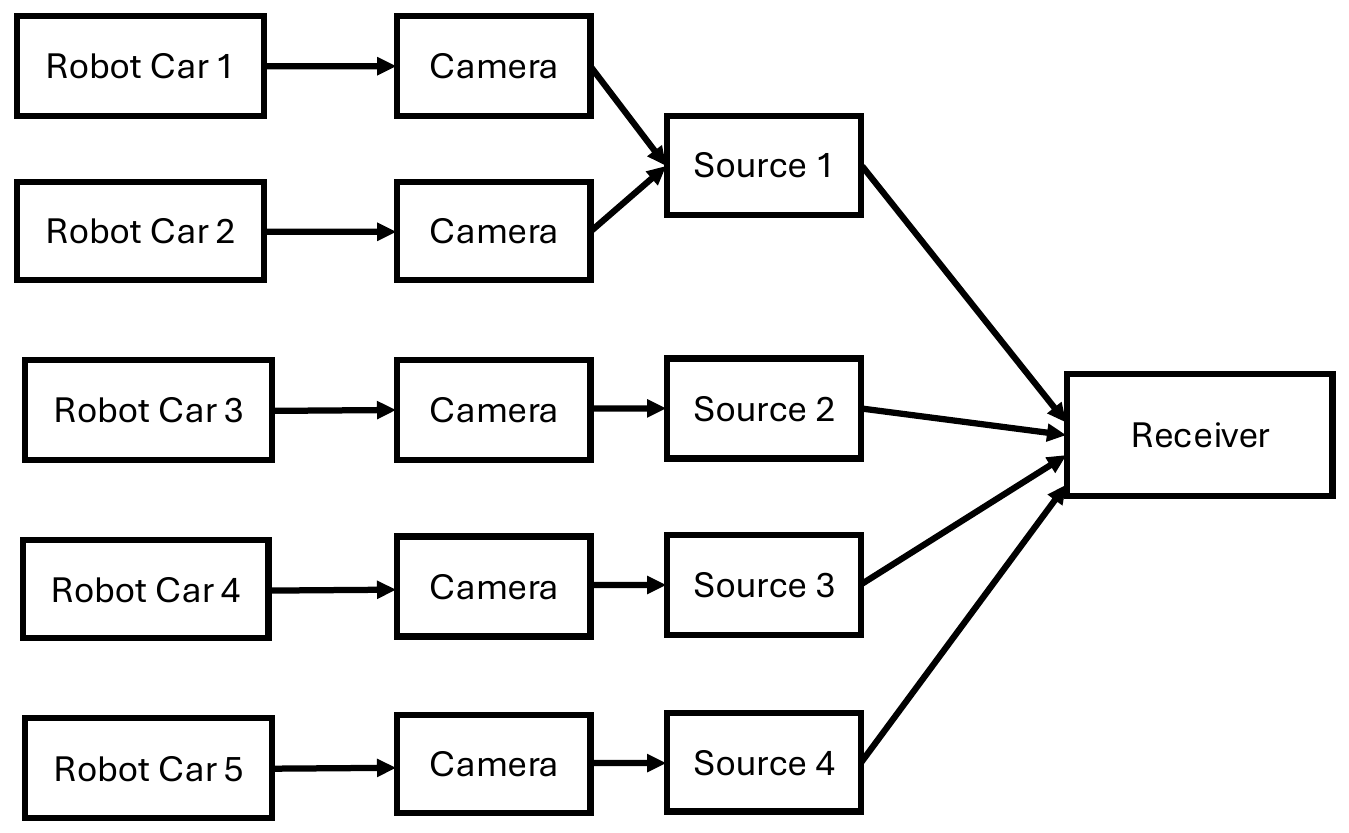}
  \caption{System Model of Experiment}\label{fig:robotmodel}
\end{figure}

\begin{figure}[t]
  \centering
\includegraphics[width=0.40\textwidth]{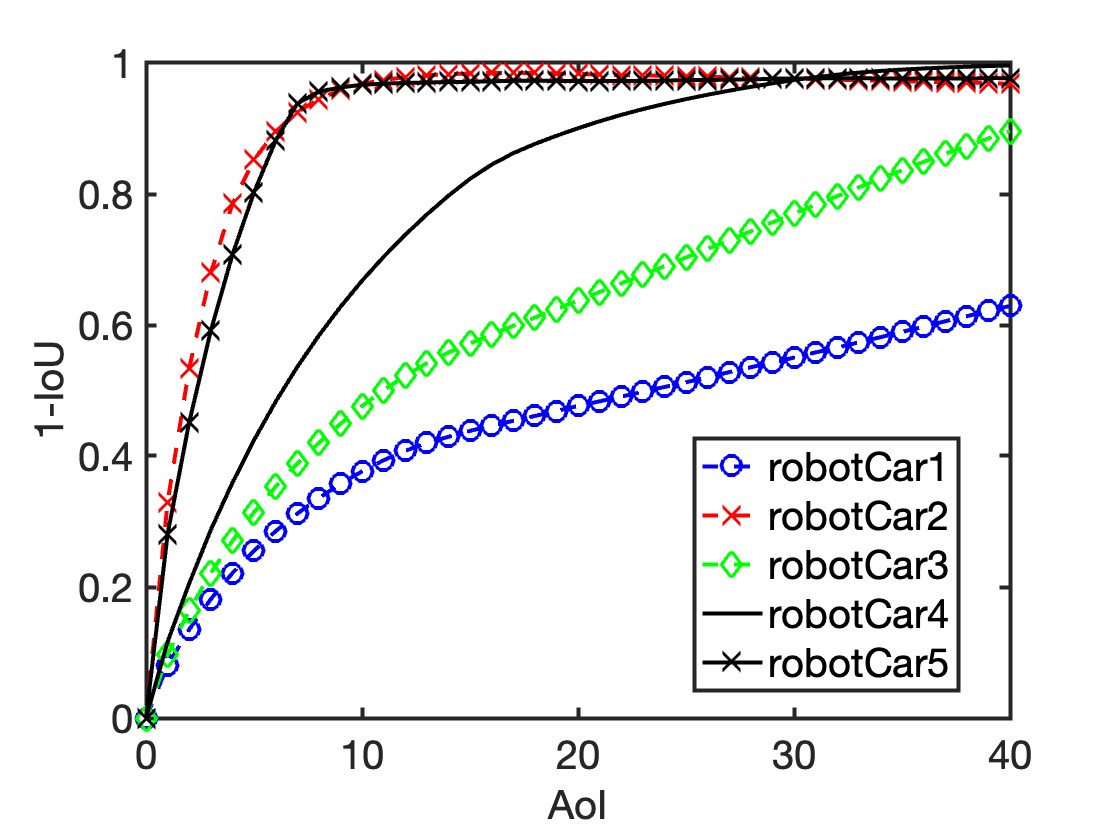}
\caption{Inference Error ($1-\mathrm{IoU}$) vs. AoI}\label{fig:roboterror}
\end{figure}

\begin{figure}[t]
  \centering
\includegraphics[width=0.40\textwidth]{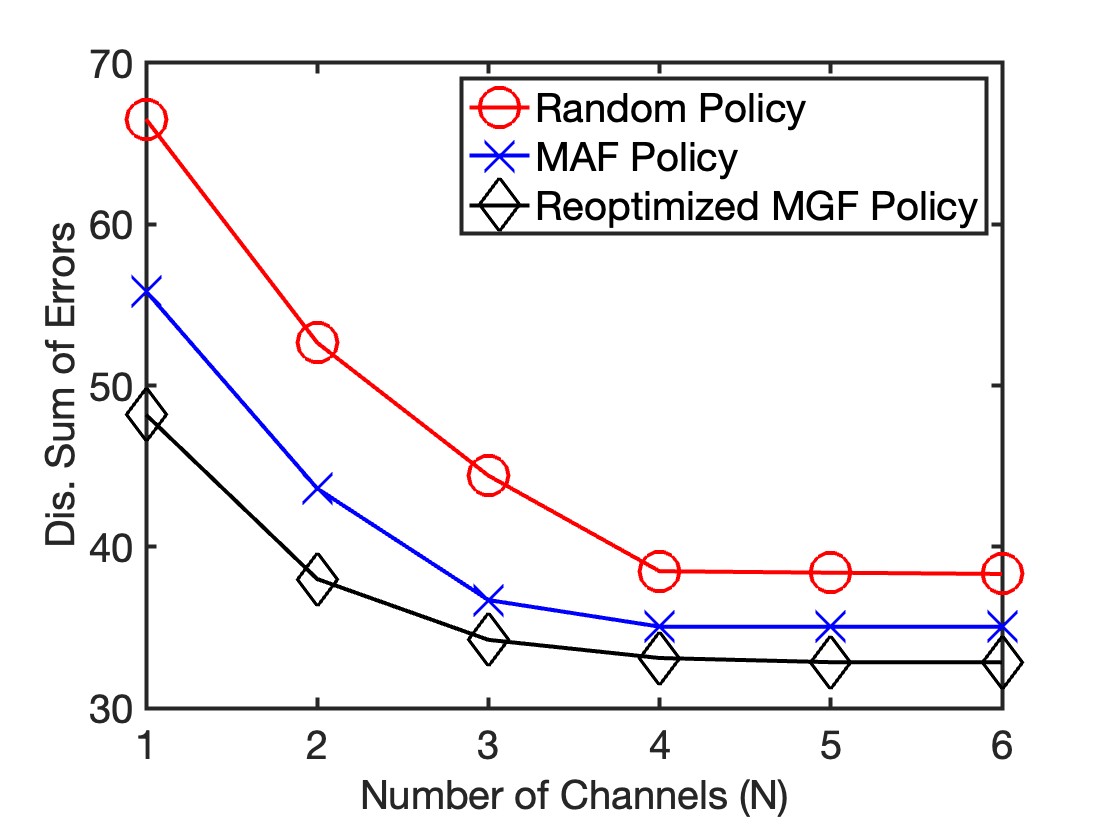}
\caption{Dis. Sum of Errors vs. $N$}\label{fig:robotchannels}
\end{figure}

\begin{figure}[t]
  \centering
\includegraphics[width=0.40\textwidth]{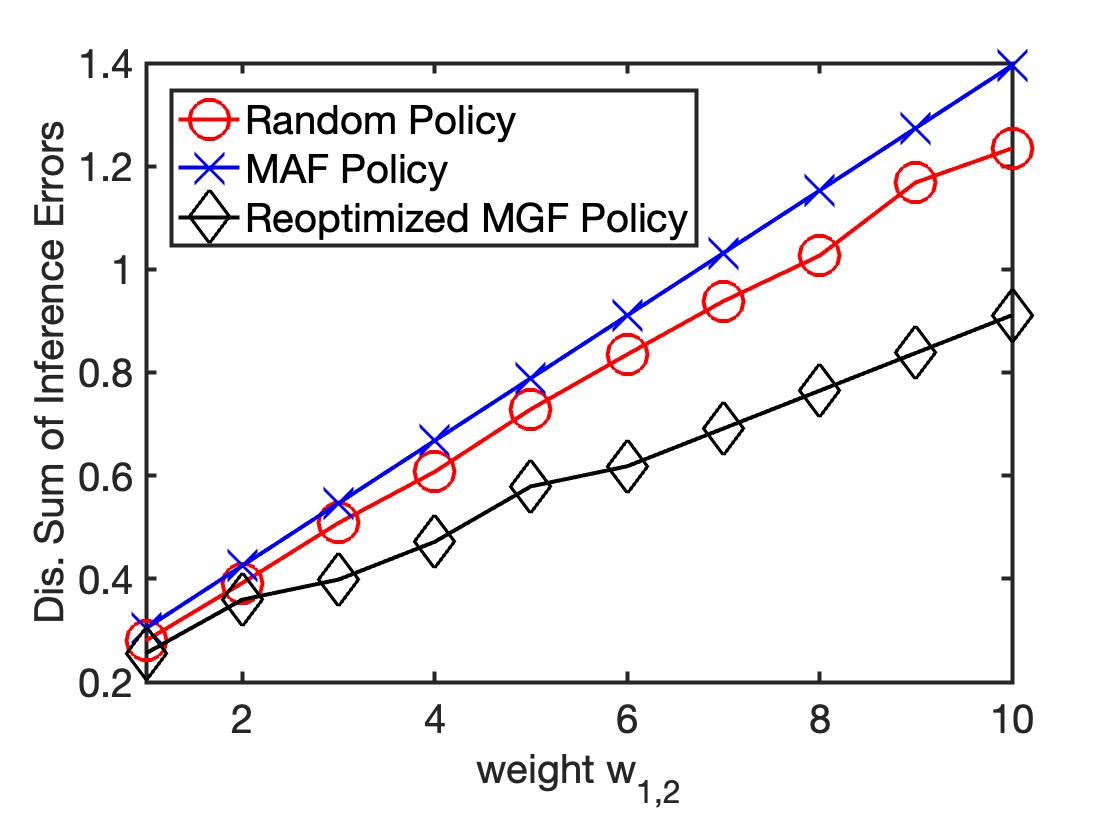}
\caption{Dis. Sum of Errors vs. $N$}\label{fig:robotweight}
\end{figure}

% \begin{figure*}[h]
%   \centering
  
%   \begin{subfigure}[t]{0.30\textwidth}
% \includegraphics[width=\textwidth]{INFOCOM2025-Kamran/RobotSystem.pdf}
%   \subcaption{System Model of Experiment}
% \end{subfigure}
% % <— this is important. There should be no empty line here. 
% %
% \hspace{1mm} 
% \begin{subfigure}[t]{0.30\textwidth}
% \includegraphics[width=\textwidth]{INFOCOM2025-Kamran/InferenceErrorRobot.jpg}
% \subcaption{\small Inference Error ($1-\mathrm{IoU}$) vs. AoI}
% \end{subfigure}
% % <— this is important. There should be no empty line here. 
% %
% \hspace{1mm} 
% \begin{subfigure}[t]{0.30\textwidth}
% \includegraphics[width=\textwidth]{INFOCOM2025-Kamran/ErrorVsNRobot.jpg}
% \subcaption{Dis. Sum of Errors vs. $N$}
% \end{subfigure}
% \vspace{-0.05in}
% \caption{\small Remote Robot Car Detection: (a) System Model of Experiments, we consider 4 sources and 5 robot cars. Source 1 observes 2 robot cars, while each of the remaining 3 sources observes a single robot car; (b) Inference Error Vs. AoI, and (c) discounted sum of errors vs. number of channels $N$. \label{fig:robot}}
% \vspace{-0.2in}
% \end{figure*}

\begin{figure*}
    \centering
    \includegraphics[width=1\linewidth]{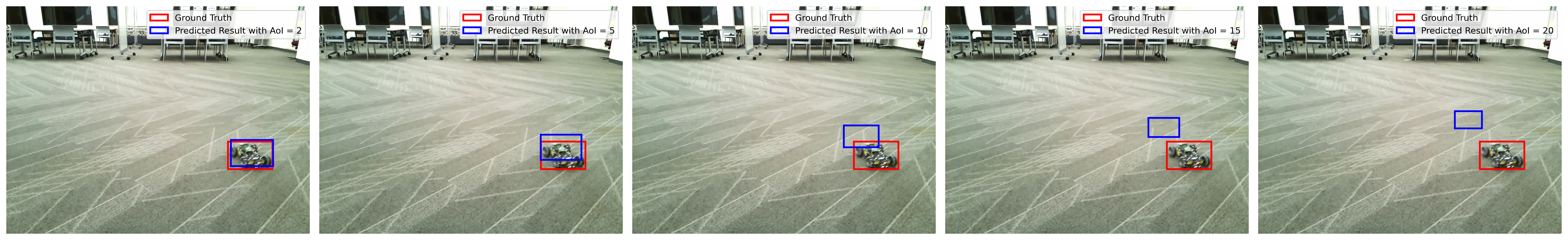}
    \caption{One sample of prediction results for robot car $3$ with different AoI values}
    \label{fig:robot3}
\end{figure*}

\begin{figure*}
    \centering
    \includegraphics[width=1\linewidth]{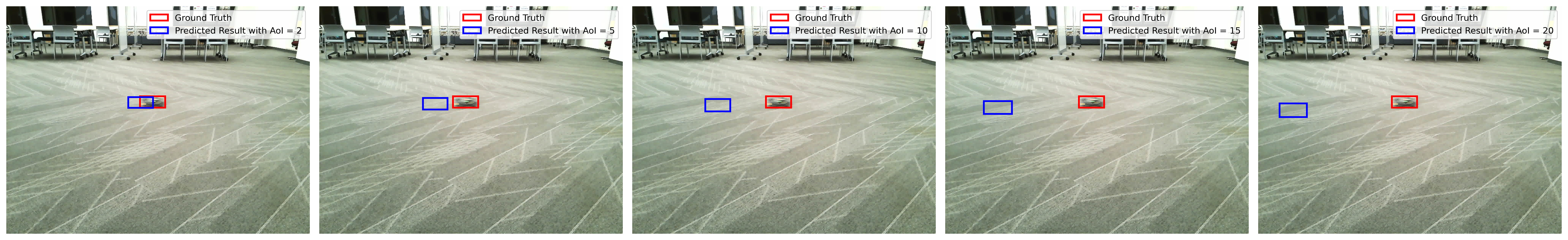}
    \caption{One sample of prediction results for robot car $2$ with different AoI values}
    \label{fig:robot2}
\end{figure*}

\subsection{Remote Robot Car Detection}\label{robotcar}

In this section, we discuss the system model and the results on remote robot car detection experiment. 

Fig. \ref{fig:robotmodel} illustrates the experimental system model, which consists of 4 sources and 5 robot cars. Source 1 observes 2 of the cars, while the remaining 3 sources each monitor a single car. Each source uses an onboard camera and the YOLO11x\cite{yolo11_ultralytics} model to detect robots within its view. The model generates bounding box information for each detection. In each time slot, the sources transmit the information of the detection to a receiver. However, the system is subject to two constraints: (i) A maximum of $N$ bounding boxes can be transmitted across the network and (ii) Source $1$ can only detect and transmit information for one of its two observed cars. The receiver's prediction for a car's location is the most recently delivered bounding box information for that car.

Fig. \ref{fig:roboterror} plots the inference error, defined as $1-\mathrm{IoU}$, versus the AoI for the $5$ robot cars. The Intersection over Union (IoU) measures the overlap between the predicted and ground-truth bounding boxes. To measure the inference error of a robot car, we have used $800$ data samples for each AoI value. One sample of prediction results for AoI values $=2, 5, 10, 15, 20$ are illustrated in Fig. \ref{fig:robot3} and Fig. \ref{fig:robot2}, respectively for robot car $3$ and $2$.  The error for robot cars 2 and 4 rapidly approaches its maximum value of 1, because these cars move much faster than the others. We can observe from Figs \ref{fig:robot3}\&\ref{fig:robot2} that robot car 2 is moving faster compared to robot car $3$ which yields larger inference error for robot car 2 compared to robot car 3.

Fig. \ref{fig:robotchannels} plots the discounted sum of the normalized inference error over $T=100$ time slots against the number of channels, $N$. For this simulation, we set the weight $w_{1,2}=5$, other weights are set to $1$, and discount factor $\gamma=0.9$. Our proposed ``Reoptimized MGF" policy clearly outperforms the ``MAF" and ``Random" baselines. The Random policy performs the worst as it does not use any state information. While the MAF policy considers the AoI, it neglects the system dynamics—specifically, how AoI impacts inference error. In contrast, our Reoptimized MGF policy leverages this system dynamics, leading to its superior performance.

Fig. \ref{fig:robotweight} plots the discounted sum of the normalized inference error over $T=100$ time slots against the weight, $w_{1,2}$ set for robot car $2$ observed by source $1$. For this simulation, other weights are set to $1$, and discount factor is set to $\gamma=0.1$. This plot illustrates that as the weight $w_{1,2}$ increases, the performance gap between our policy and other baselines increases. 

\subsection{Traffic Prediction and Segmentation}\label{machinelearning}

  \begin{figure}[t]
    \centering
\includegraphics[width=0.40\textwidth]{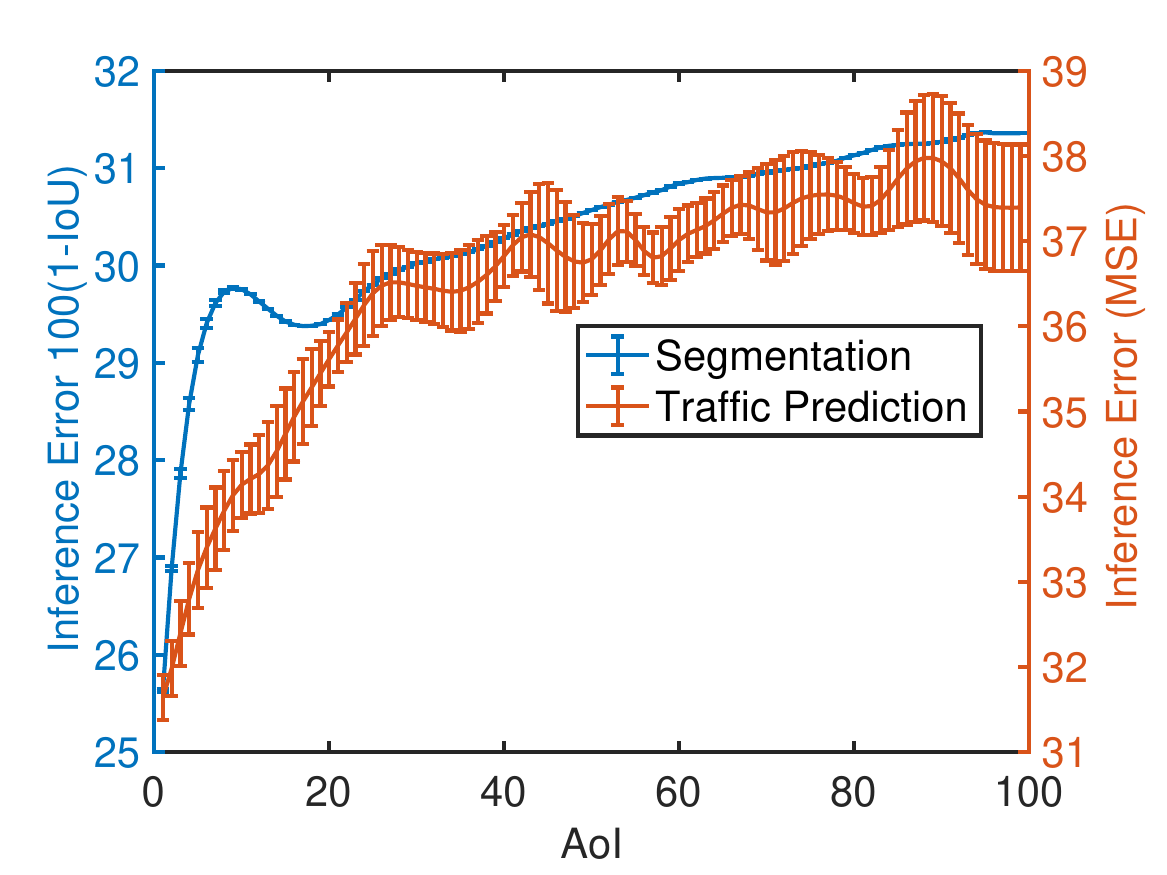}
  \caption{Inference Error vs AoI}\label{fig:traffic}
\end{figure}

\begin{figure}[t]
    \centering
\includegraphics[width=0.40\textwidth]{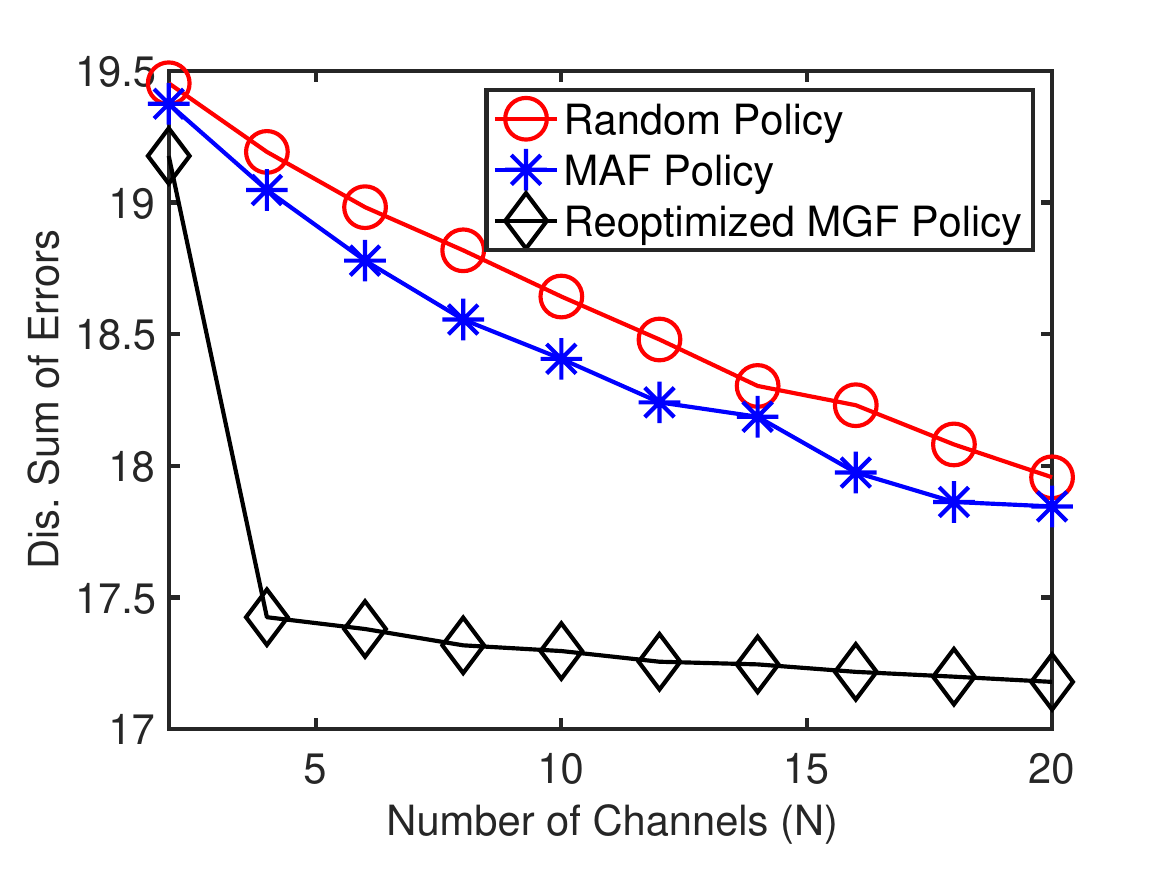}
\caption{Dis. Sum of Errors vs. $N$}\label{fig:trafficchannel}
\end{figure}

\begin{figure}[t]\centering
\includegraphics[width=0.40\textwidth]{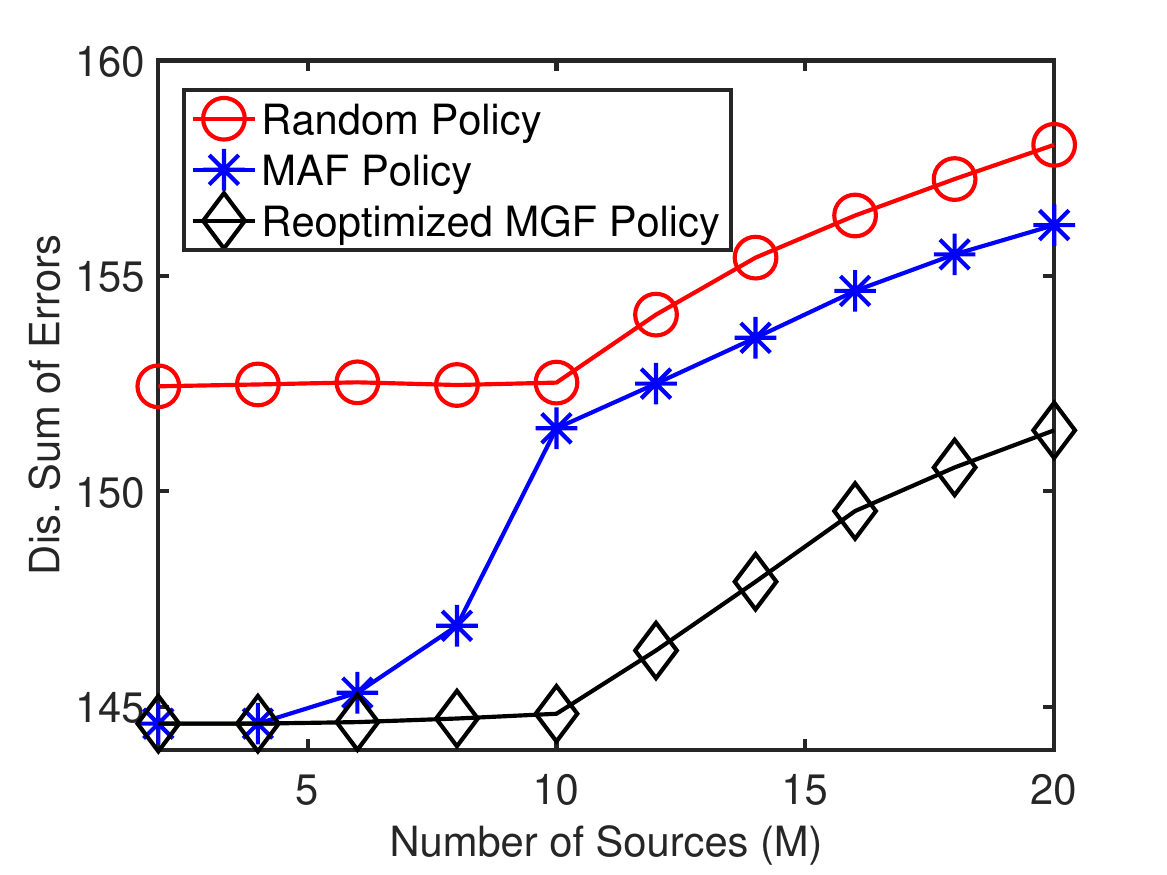}
\caption{Dis. Sum of Errors vs. $M$}\label{fig:trafficsources}
\end{figure}

% \begin{figure*}[h]
%   \centering
  
%   \begin{subfigure}[t]{0.30\textwidth}
% \includegraphics[width=\textwidth]{MLexp.eps}
%   \subcaption{Inference Error vs AoI}
% \end{subfigure}
% % <— this is important. There should be no empty line here. 
% %
% \hspace{1mm} 
% \begin{subfigure}[t]{0.30\textwidth}
% \includegraphics[width=\textwidth]{ErrorVsChannel1.eps}
% \subcaption{\small Dis. Sum of Errors vs. $N$}
% \end{subfigure}
% % <— this is important. There should be no empty line here. 
% %
% \hspace{1mm} 
% \begin{subfigure}[t]{0.30\textwidth}
% \includegraphics[width=\textwidth]{ErrorVsSources.eps}
% \subcaption{Dis. Sum of Errors vs. $M$}
% \end{subfigure}
% \caption{\small Data-driven evaluations: (a) Inference error vs. AoI for segmentation task (left), where we use 100~(1-IoU) loss function, and for traffic prediction (right), where we use quadratic loss function. (b) Discounted sum of errors vs. number of channels $N$ in which two inference tasks are given priority with weight $1$ and the others have weight $0.01$. (c) Discounted sum of errors vs. number of sources $M$ in which half of the inference tasks given priority with weight $1$ and the others have weight $0.01$. MGF outperforms the other two policies. 
% %\chris{Add key takeaway.}
% \label{fig:learning}}
% \vspace{-0.2in}
% \end{figure*}

We consider two machine learning tasks: (i) scene segmentation and (ii) traffic prediction. To collect the inference error functions, we employ the NGSIM dataset \cite{NGSIM_I80_2016, NGSIM_US101_2016, NGSIM_Lankershim_2016, NGSIM_Peachtree_2016}, %worth citing the dataset
which contains video recordings from roadside unit (RSU) cameras installed above four different US road surfaces. These videos capture traffic from various camera angles around the road surfaces and were recorded at different times of the day, each for a duration of 15 minutes.

In our experiments, each source is modeled as an RSU. Each RSU generates features for the two inference tasks: video frame segmentation and traffic prediction. We define a time slot as the duration of two video frames, during which feature generation and transmission are completed. We randomly select $4$ videos for our analysis.

%To begin the experiments, we simulated the data for two distinct machine learning scenarios:
%\begin{itemize}
 %   \item Image Segmentation - segmenting all frames of a given video into distinct objects.
  %  \item Vehicular Detection - detecting and counting the number of vehicles in the frame.
%\end{itemize}
%The simulated data from these scenarios was later used in the frame prediction framework, where a model tried to predict the future frames based on a historical frame window. 

{\bf Scene Segmentation:} For image segmentation, we utilize the Segment Anything Model (SAM) released by Meta AI \cite{kirillov2023segany}.
We adopt the medium ViT-L model as a pre-trained model to segment each frame into distinct areas. We split the ViT-L model into two parts: a feature generator and a predictor. 
%We split ViT-L model to a feature generator model and a predictor model. 
%Along the segmentation task we run a prediction task in which the past segmented frames are used to predict segmentation of current video frame. 
The predictor model takes feature generated at time $t-\delta$ as input to predict the segmentation for frame at time $t$, where $\delta$ is the AoI value. By taking feature produced at time $t$, we generate ground truth for loss calculation. 
%the frame segmentation is generated using the ViT-L model and is then utilized as a ground truth to evaluate the prediction quality. 
We employ the Intersection over Union 100(1-IoU) loss metric, where $\text{IoU}=\frac{\hat{S} \cap S}{\hat{S} \cup S}$, $S$ is the ground truth segmentation of the frame containing combined masks for all distinct segments, and $\hat{S}$ is the predicted segmented frame. We use the loss function over the selected videos from the dataset to generate inference error.

%Specifically, the output of $t-\delta$ ago to predict the segmentation of the frame at time $t$. 
% The loss function we use is the Intersection over Union (IoU) loss function calculated as 
%{\red Discuss how to calculate the loss function}

%The segmented frames were then utilized to predict the future segmetnation in an online manner. The model we utilized was using the $t-k$ segmented frame to predict the segmentation for frame $t$ for different values of $k$.

{\bf Traffic Prediction:}
%{\red Discuss how to get the features and how to get the predictor. For feature generation, we are using pre-trained models.} 
For traffic prediction, 
%we generate simulated data sequences containing vehicle counts and positions for each video frame. 
we leverage image pre-processing techniques and pre-trained state-of-the-art (SOTA) models. Each frame is duplicated, with one copy undergoing SAM-based segmentation mask application and the other undergoing grayscaling, edge enhancement, resizing, and blurring. Both processed frames are then fed into a pre-trained YOLOv8 \cite{Jocher_Ultralytics_YOLO_2023} image detection model to identify all vehicles. The detected vehicles from both frames are combined, removing any overlaps, and their positions are saved, creating the final data sequences. Combined SAM-YOLOv8 model, along with pre-processing, serves as feature generator. 

After generating the data sequences, we split them into 80\% training and 20\% inference datasets. Our prediction framework utilizes a separate LSTM model for each AoI value $\delta$, with hyperparameters detailed in Table \ref{tab}. For a given AoI $\delta$, the input to the corresponding LSTM model is the sequence of vehicle counts from $\delta - l$ to $\delta$ frames ago, with the goal of predicting the current vehicle count. We train each model for $50$ epochs. %evaluating the prediction performance after each epoch using the quadratic loss between the LSTM output and the ground truth vehicle count at frame $t$. 
Using the trained LSTM models, we record inference errors for each AoI $\delta =1, 2, \dots, 100$.
%Fig. \ref{fig:learning}(a) depicts inference error function for traffic prediction with variance.

% Add this to your preamble

\begin{table}[htbp]
\centering
\begin{tabular}{lr}
\toprule
\textbf{Hyperparameter} & \textbf{Value} \\
\midrule
Hidden units per layer & 16 \\
Input and output dimensions & 1 \\
Batch size & 32 \\
Window size ($l$) & 3 \\
Optimizer & Adam \\
Learning rate & 0.0001 \\
\bottomrule
\end{tabular}
\caption{Hyperparameters used for the LSTM models.}
\label{tab}
\end{table}

% \begin{table}[h]
% \centering
% \begin{tabular}{c|c}
% \textbf{Hyperparameter} & \textbf{Value} \\
% \hline
% Hidden Layers & 16 \\
% \hline
% Input and Output Size & 1 \\
% \hline
% Batch Size & 32 \\
% \hline
% Window Size ($l$) & 3 \\
% \hline
% Optimizer & Adam \\
% \hline
% Learning Rate & 0.0001 \\
% \hline
% \end{tabular}
% \caption{Hyperparameters used for the LSTM models.}
% \label{tab}
% \end{table}

Fig. \ref{fig:traffic} illustrates the resulting inference errors vs. AoI.

We now evaluate the scheduling policies employing these inference error functions. In Fig. \ref{fig:trafficchannel}, we plot the discounted sum of errors against the number of channels $N$ over a time horizon of $T=100$, with two inference tasks $k_m=2$ per source and scaling factor $r=1$. We set $M=20$, $\gamma=0.9$, $n_{m,j}=1$, and $C_m=2$ for all sources. Task weights $w_{m,j}$ are set to 1 for tasks (1,2) and (5,1), and 0.01 for the rest. As expected, increasing $N$ is seen to improve performance across all policies. Notably, when $N=4$, the MGF policy outperforms the MAF policy by $10\%$. Additionally, Fig. \ref{fig:trafficchannel}) clearly demonstrates the consistently poor performance of the random policy.

Fig. \ref{fig:trafficsources} illustrates the performance of the scheduling policies as the number of sources $M$ increases, over a finite horizon of $T=100$. Each source has two inference tasks $k_m=2$ and $r=1$, with other simulation parameters set to $N=10$, $\gamma=0.9$, $n_{m,j}=1$, and $C_m=1$. We assign weights $w_{m,j}$ of $1$ to half the inference tasks and $0.01$ to the rest. We see that while the MAF and MGF policies perform similarly with a small number of sources, the MGF policy becomes better as the number of sources increases.

\section{Conclusion}

In this paper, we studied the computation and communication co-scheduling problem in MTRI systems to minimize inference errors under resource constraints. We formulated this problem as a weakly-coupled MDP with inference errors described as penalty functions of AoI. To address the resulting PSPACE-hard complexity, we developed a novel reoptimized MGF policy, which our theoretical analysis proved to be asymptotically optimal as the number of inference tasks increases. We also discussed how to simplify a reoptimized policy by reducing the number of optimization variables. 
Numerical evaluations using both synthetic and real-world datasets further validated our reoptimized MGF's superior performance compared to baseline policies. For synthetic evaluations, we use three different types of inference error functions that are widely used in AoI literature. For real-world experiments, remote robot car detection and  vehicular inference tasks are studied.  

\begin{appendices}

%\begin{algorithm}[t]
%\SetAlgoLined
%\SetAlgoNoEnd
%\SetKwInOut{Input}{input}
%\caption{Backward Induction for Value Function}\label{alg:backwardinduction}
%\begin{algorithmic}[1]
%$V^{\boldsymbol{\lambda}_m, \boldsymbol{\mu}}_{m, j, t}(b, \delta)\gets 0$ for all $b$ and $\delta$\\
%\For{$k=T, T-1, \ldots, 0$}
%{Compute $V^{\boldsymbol{\lambda}_m, \boldsymbol{\mu}}_{m, j, k}(b, \delta)$ using \eqref{valuefunction}
%for all $1 \leq b \leq \bar b$ and $1 \leq \delta \leq \bar \delta$
%}
%\end{algorithmic}
%\end{algorithm}

%\section{Proof of Theorem \ref{theorem1}}\label{ptheorem1}
%Given $\Delta_{m,j}(t) = \delta$, the optimal policy at time $t$ satisfies the Bellman optimality equation \cite[Proposition 1.3.1]{bertsekasdynamic1}: 
%\begin{align}\label{Optimality}
%V^{\boldsymbol{\lambda}_m, \boldsymbol{\mu}}_{m, j, t}(\delta)=&\min_{a \in \{0,1\}}r(\delta,a)+\gamma\mathbb E[V^{\boldsymbol{\lambda}_m, \boldsymbol{\mu}}_{m, j, t+1}(\Delta_{m,j}(t+1)],
%\end{align}
%where $r(\delta,a)$ is the immediate cost at time $t$ for any action $\pi_{m,j}(t)=a$. The immediate cost comprises the weighted inference error $w_{m,j} p_{m,j}(\delta)$, a potential cost $\lambda_m+\mu_m$ if $\pi_{m,j}(t) = 1$. Substituting immediate cost and state transition dynamics \eqref{AoIprocess} into the optimality equation \eqref{Optimality}, we obtain the value function update \eqref{valuefunction}. This completes the proof.

\section{Proof of Lemma \ref{lemmatheorem4}}\label{plemmatheorem4}

Firstly, we denote the action at time $t$ for $j$-th inference task of $m$-th source under the reoptimized MGF policy provided in Algorithm \ref{alg:gain} by $\pi^{\text{MGF}}_{m, j}(t)$. We denote the reoptimized action at time $t$ for $j$-th inference task of $m$-th source under the policy $\pi^*$ by $\pi^*_{m, j}(t)$, where $\pi^*$ maximizes $\bar p_t(\boldsymbol{\lambda}_{t:T}, \boldsymbol{\mu}_{t:T})$ defined in \eqref{dualProblem}.

Next, we denote the number of subproblems with actions that differ between the reoptimized MGF policy provided and the policy $\pi^*$ at time $t$ by $I_t$:
$$I_t=\bigg\{(m, j): \pi^{\text{MGF}}_{m, j}(t) \neq \pi^*_{m, j}(t)\bigg\},$$
\begin{align}
    C^*_m(t)=\sum_{j=1}^{k_m} \pi^*_{m,j}(t),
    N^*(t)=\sum_{m=1}^M\sum_{j=1}^{k_m} \pi^*_{m,j}(t),
\end{align}
where we use $n_{m, j}=1$ for the simplicity of analysis in this proof.

{\bf Case 1}: At time $t$, all constraints are satisfied under policy $\pi^*$. 
In this case, we have $|I_t|=0.$
{\bf Case 2}: At least one constraint does not satisfy under policy $\pi^*$. 
{\blue In this case, if a sub-problem $(m, j) \in I_t$, then $\pi^*_{m, j}(t)=1$ and $\pi^{\text{MGF}}_{m, j}(t)=0$ due to resource limitation, i.e., $C^*_m(t)>C_m$ or $N^*(t)>N$ or both. Because the active action $\pi^*_{m, j}(t)=1$ consumes one communication and one computation resource, we can upper bound $I_t$ by
\begin{align}\label{I2case}
    |I_t| \leq \sum_{m=1}^M (C^*_m(t)-C_m)^{+}+(N^*(t)-N)^{+}.
\end{align}}
%where $\sum_{m=1}^M (C^*_m(t)-C_m)^{+}$ indicates the number of sub-problems can not be activated due to computation resource limitation and $(N^*(t)-N)^{+}$ indicates the number of sub-problems can not be activated due to communication resource limitation.}  

%Combining {\bf Case 1}, {\bf Case 2}, and 
By taking average over all possible AoI values, we have 
\begin{align}\label{C_m}
    &\mathbb E[(C^*_m(t)-C_m)^{+}]^2\nonumber\\
    &\overset{(a)}{\leq} \mathbb E[(C^*_m(t)-\mathbb E[ C^*_m(t)])^{+}]^2 \nonumber\\
    &\overset{(b)}{\leq} \text{Var}(C^*_m(t))\nonumber\\
    &\overset{(c)}{\leq} \sum_{j=1}^{k_m} \bigg(E[\pi^*_{m,j}(t)]-\mathbb E[\pi^*_{m,j}(t)]^2\bigg)\nonumber\\
    &{\leq} k_m, 
\end{align}
where (a) holds because on average $\mathbb E[C^*_m(t)]\leq C_m$, see \cite[Proposition 3.2(c)]{brown2023fluid}, (b) holds due to Jensen's inequality, (c) is because of Bhatia-Davis inequality.
%and (d) follows due to 
%\begin{align}
 %   0\leq \bigg(1-\mathbb E[\bar \pi_{m,j}(t)]\bigg)\mathbb E[\bar \pi_{m,j}(t)]\leq 1.
%\end{align}
Similarly, we can have 
\begin{align}\label{N_m}
    &\mathbb E[(N^*(t)-N)^{+}]^2 \nonumber\\
    &\leq \mathbb E[(N^*(t)-\mathbb E[N^*])^{+}]^2\nonumber\\
    &\leq \text{Var}(N^*(t))\nonumber\\
    &\leq \sum_{m=1}^M k_m.
\end{align}
By taking an average on \eqref{I2case} and substituting \eqref{C_m} and \eqref{N_m} into \eqref{I2case}, we obtain
\begin{align}\label{Ilast}
    \mathbb E[|I_t|]&=\sum_{m=1}^M \mathbb E[(C^*_m(t)-C_m)^{+}]+\mathbb E[(N^*(t)-N)^{+}]\nonumber\\
    &\leq \sum_{m=1}^M \sqrt{k_m}+\sqrt{\sum_{m=1}^M k_m}.
\end{align}
This concludes the proof. 

\section{Proof of Theorem \ref{theorem4}}\label{ptheorem4}
To prove this theorem, we begin with a definition of re-optimized fluid (ROF) policy \cite{brown2023fluid}. Leveraging Propositions 3.2 and 3.4 of \cite{brown2023fluid}, which establish the equivalence of optimal actions under dynamic fluid and Lagrangian relaxations, we define the re-optimized fluid (ROF) policy: 
%using the optimal actions from our Lagrangian relaxed problem \eqref{dual}. 

\begin{definition}[\textbf{Re-optimized Fluid Policy}\cite{brown2023fluid}]\label{ROFP} 
%We define the set of ROF policies for our problem \eqref{Multi-scheduling_problem}-\eqref{Sceduling_constraint2}. 
Any reoptimized feasible fluid policy $\pi$ up to a finite time $T$ satisfies:
\begin{itemize}[leftmargin=4mm]
    \item At every time $t$, the policy updates $\Delta_{m,j}(t)$ and generates an action $\pi^*_{m,j}(t)$ independently across all sub-problems that is optimal to \eqref{dual} with optimal Lagrange multipliers. 
    \item Assigns $\pi_{m,j}(t)=0$ for all $(m, j)$. Then, in any pre-defined order among all sub-problems $(m, j)$, update action $\pi_{m,j}(t)=\pi^*_{m,j}(t)$ if all constraints are satisfied. {\blue In this paper, we employ maximum gain index first strategy for ordering the sub-problems.}    
\end{itemize}
\end{definition}

Algorithm \ref{alg:gain} and Definition \ref{ROFP} implies that MGF belongs to ROF policies. The ROF policies are proven to be asymptotically optimal \cite{brown2023fluid}. We prove Theorem \ref{theorem4} for our problem with tighter bound than that established in \cite{brown2023fluid}. Firstly, we omit $r$ for the simplicity of presentation. We use $n_{m, j}=1$ for the simplicity of analysis.  

Because $p_{m, j}(\delta)$ is bounded, there exist finite constants $\bar{p}_h$ and $\bar{p}_l$ such that 
$\bar{p}_l \leq w_{m, j} p_{m, j}(\delta) \leq \bar{p}_h$. Let $\bar{p}_{opt}(T)$ and $\bar{p}_{{\text{MGF}}}(T)$ denote the discounted sum of inference errors under an optimal policy to \eqref{Multi-scheduling_problem}-\eqref{Sceduling_constraint2} and the MGF policy, respectively, {\blue truncated at} time $T$. Then, we have
\begin{align}\label{main}
&\bar{p}_{{\text{MGF}}} - \bar{p}_{opt}\nonumber\\
%&\overset{(a)}{\leq} \bar{p}_{{\text{MGF}}}(T) - \bar{p}_{opt}(T) + \frac{1}{K}\sum_{t=T+1}^{\infty} \gamma^t \sum_{m=1}^M \sum_{j=1}^{k} (\bar{p}_h - \bar{p}_l) \nonumber\\
%&= \bar{p}_{{\text{MGF}}}(T) - \bar{p}_{opt}(T) + \frac{\gamma^{T+1}(\bar{p}_h - \bar{p}_l)}{(1-\gamma)K} K \nonumber\\
&{\leq} ~\bar{p}_{{\text{MGF}}}(T) - \bar{p}_{1}(\boldsymbol{\lambda}^*_{1:T}, \boldsymbol{\mu}^*_{1:T}) + \frac{\gamma^{T+1}(\bar{p}_h - \bar{p}_l)}{1-\gamma},
\end{align}
where the inequality holds because the penalty functions are bounded and the weak duality
$\bar{p}_{1}(\boldsymbol{\lambda}^*_{1:T}, \boldsymbol{\mu}^*_{1:T}) \leq \bar{p}_{opt}(T)$.

%The policy $\pi^*=(\pi^*_{m,j}(t))_{\forall m,j, t\leq T}$ is the optimal to \eqref{dual}, where  the action $\pi^*_{m, j}(t)$ is obtained by using Lemma \ref{lemma1} with optimal Lagrange multipliers of \eqref{dualProblem}. Using \eqref{optimalsubpolicy} and \eqref{gain}, we can verify that $\pi^*_{m, j}(t) = 1$ if and only if 
%$\alpha_{m,j,t}(\Delta_{m,j}(t))>0$.

%\begin{lemma}\label{lemmatheorem4}
%    For any time $t$, the expected number of inference tasks $(m, j)$ with different actions under the MGF policy and the policy $\pi^*$ is bounded above by $\sum_{m=1}^M \sqrt{k_m}+\sqrt{\sum_{m=1}^M k_m}$.
%\end{lemma}  
%\begin{proof}
 %   See Appendix \ref{plemmatheorem4}.
%\end{proof}

%Lemma \ref{lemmatheorem4} differs from Lemma EC1.1 in the Appendix of \cite{brown2023fluid} due to the linear increase in computation resource constraints with the number of sources $M$ in our formulation, whereas resource constraints are constant in \cite{brown2023fluid}.

{\blue Let $B_t$ denote the expected number of inference tasks $(m, j)$ with different actions under the MGF policy and the policy $\pi^*$. We have 
\begin{align}\label{eqnmax}
    B_t &\overset{a}{\leq} \sum_{m=1}^M \sqrt{k_m}+\sqrt{\sum_{m=1}^M k_m} \overset{b}{\leq} 2\sum_{m=1}^M \sqrt{k_m}, 
\end{align}
where $(a)$ holds due to Lemma \ref{lemmatheorem4} and (b) holds  because $\|\mathbf x\|_2 \leq \| \mathbf x\|_1$ for the vector $\mathbf x=[\sqrt k_1, \sqrt k_2, \ldots, \sqrt k_m]$.

Similar to \cite[corollary 4.4]{brown2023fluid}, we can show the following Lemma:
\begin{lemma}\label{lemma3}
For our re-optimized fluid policy, we have  
\begin{align}
    \bar{p}_{{\text{MGF}}}(T) - \bar{p}_{1}(\boldsymbol{\lambda}^*_{1:T}, \boldsymbol{\mu}^*_{1:T})\leq \frac{\gamma (\bar p_h-\bar p_l) max_{t} B_t}{(1-\gamma)^3 \sum_{m=1}^Mk_m}.
    \end{align}
\end{lemma}
By using similar proof steps provided by \cite{brown2023fluid}, we can prove Lemma \ref{lemma3}. By combining \eqref{eqnmax} and Lemma \ref{lemma3}, we can establish 
\begin{align}\label{anotherbound}
    \bar{p}_{{\text{MGF}}}(T) -\bar{p}_{1}(\boldsymbol{\lambda}^*_{1:T}, \boldsymbol{\mu}^*_{1:T})&\leq  \frac{2(\bar p_h-\bar p_l)\gamma \sum_{m=1}^M \sqrt{k_m}}{(1-\gamma)^3 \sum_{m=1}^M k_m}\nonumber\\
    &\leq  \frac{2M(\bar p_h-\bar p_l)\gamma \sum_{m=1}^M \sqrt{k_m}}{(1-\gamma)^3 (\sum_{m=1}^M \sqrt{k_m})^2}\nonumber\\
    &\leq  \frac{2M(\bar p_h-\bar p_l)\gamma}{(1-\gamma)^3 \sum_{m=1}^M \sqrt{k_m}},
   % =O\left(\frac{1}{\sum_{m=1}^M \sqrt{k_m}}\right),
\end{align}
where the second inequality holds due to $\frac{1}{\|\mathbf x\|_2^2}\leq \frac{M}{\|\mathbf x\|_1^2}$.
%where the differences from \cite{brown2023fluid} is the upper bound of Lemma 2 in our paper and we take average over all inference tasks.

By substituting \eqref{anotherbound} and $T = \log_{\frac{1}{\gamma}}\sum_{m=1}^M \sqrt{k_m}$ into \eqref{main}, we obtain 
\begin{align}
&\bar{p}_{{\text{MGF}}} - \bar{p}_{opt}\leq \frac{2M(\bar p_h-\bar p_l)\gamma}{(1-\gamma)^3 \sum_{m=1}^M \sqrt{k_m}} + \frac{(\bar p_h-\bar p_l)\gamma}{(1-\gamma)\sum_{m=1}^M \sqrt{k_m}}\nonumber\\
&\leq \frac{1}{\sum_{m=1}^M \sqrt{k_m}} 
\left(\frac{2M(\bar p_h-\bar p_l)\gamma}{(1-\gamma)^3} + \frac{(\bar p_h-\bar p_l)\gamma}{(1-\gamma)}\right).
\end{align}}
By substituting $k_m=rk_m$, $C_m=rC_m$, $N=rN$, and maintaining $M$ sources and $k_m$ class of sub-problems constant, we arrive at Theorem \ref{theorem4}. Note that changing $M$ sources and $k_m$ class of sub-problems would alter the optimal Lagrange multipliers. 
%necessitating their constancy while increasing the number of inference tasks $r$ per class of sub-problems. 
This concludes the proof.

% {\red 1. one device one sensor to multiple tasks. How about if we generalize to both ``one device one sensor to multiple tasks" and ``one device multiple sensors to multiple tasks", where each task is associated with one sensor. Fig. 1 can have two parts: one for single device and then we use another figure to show multiple devices}

% {\red 2. Currently, we consider device binary decision, where if $\pi_{m, j}=1$, then the $m$-th task of $j$-th device is generated and sent. Can we extend it to the multiple action problem, where $a_{m, j}$ denotes feature generation and $u_{m, j}$ denotes transmission decision. Whenever a feature is generated, it is stored in a buffer. Whenever, a task is selected for transmission, the feature from the buffer is sent. Can we have the theoretical analysis??}

% {\red 3. For algorithm simplification, we consider time-independent lagrange multiplier, i.e., $\lambda_{m, t}=\lambda_m$. Then, we optimize $\lambda_m$ at every time.}

% {\red 4. We will contribute dataset from our experiment and we will provide more details on feature generator function and predictor function applied in our experiment}

% {\blue Plan: I will help Adam to work on experiment setup, dataset generation, and generating inference error function. I will look into the simulation codes for scheduling, the theoretical contribution, and paper writing. Which journal TMC or ToN??}

\end{appendices}

\bibliographystyle{IEEEtran}
\bibliography{refshisher}

\end{document}